\documentclass[12pt, draftclsnofoot, onecolumn]{IEEEtran}

\usepackage{fancyhdr}
\usepackage{epsfig}
\usepackage{threeparttable}
\usepackage{epsf,epsfig}
\usepackage{amsthm}
\usepackage[cmex10]{amsmath}
\usepackage{amssymb, amsfonts}
\usepackage[noadjust]{cite}
\usepackage{dsfont}
\usepackage{subfigure}
\usepackage{color}
\usepackage{soul}
\usepackage{amssymb}
 \usepackage{accents}
  \usepackage{algorithm}
   \usepackage[english]{babel}
    \usepackage{url}
    \usepackage{framed}

\begin{document}
\newtheorem{theorem}{Theorem}
\newtheorem{acknowledgement}[theorem]{Acknowledgement}
\newtheorem{axiom}[theorem]{Axiom}
\newtheorem{case}[theorem]{Case}
\newtheorem{claim}[theorem]{Claim}
\newtheorem{conclusion}[theorem]{Conclusion}
\newtheorem{condition}[theorem]{Condition}
\newtheorem{conjecture}[theorem]{Conjecture}
\newtheorem{criterion}[theorem]{Criterion}
\newtheorem{definition}{Definition}
\newtheorem{exercise}[theorem]{Exercise}
\newtheorem{lemma}{Lemma}
\newtheorem{corollary}{Corollary}
\newtheorem{notation}[theorem]{Notation}
\newtheorem{problem}[theorem]{Problem}
\newtheorem{proposition}{Proposition}
\newtheorem{scheme}{Scheme}   
\newtheorem{protocol}{Protocol}   
\newtheorem{solution}[theorem]{Solution}
\newtheorem{summary}[theorem]{Summary}
\newtheorem{assumption}{Assumption}
\newtheorem{example}{\bf Example}
\newtheorem{remark}{\bf Remark}

\def\qed{$\Box$}
\def\QED{\mbox{\phantom{m}}\nolinebreak\hfill$\,\Box$}
\def\proof{\noindent{\emph{Proof:} }}
\def\poof{\noindent{\emph{Sketch of Proof:} }}
\def
\endproof{\hspace*{\fill}~\qed
\par
\endtrivlist\unskip}
\def\endproof{\hspace*{\fill}~\qed\par\endtrivlist\vskip3pt}

\def\E{\mathsf{E}}
\def\eps{\varepsilon}
\def\phi{\varphi}
\def\Lsp{{\boldsymbol L}}
\def\Bsp{{\boldsymbol B}}
\def\lsp{{\boldsymbol\ell}}
\def\Ltsp{{\Lsp^2}}
\def\Lpsp{{\Lsp^p}}
\def\Linsp{{\Lsp^{\infty}}}
\def\LtR{{\Lsp^2(\Rst)}}
\def\ltZ{{\lsp^2(\Zst)}}
\def\ltsp{{\lsp^2}}
\def\ltZt{{\lsp^2(\Zst^{2})}}
\def\ninN{{n{\in}\Nst}}
\def\oh{{\frac{1}{2}}}
\def\grass{{\cal G}}
\def\ord{{\cal O}}
\def\dist{{d_G}}
\def\conj#1{{\overline#1}}
\def\ntoinf{{n \rightarrow \infty }}
\def\toinf{{\rightarrow \infty }}
\def\tozero{{\rightarrow 0 }}
\def\trace{{\operatorname{trace}}}
\def\ord{{\cal O}}
\def\UU{{\cal U}}
\def\rank{{\operatorname{rank}}}
\def\acos{{\operatorname{acos}}}

\def\SINR{\mathsf{SINR}}
\def\SNR{\mathsf{SNR}}
\def\SIR{\mathsf{SIR}}
\def\tSIR{\widetilde{\mathsf{SIR}}}
\def\Ei{\mathsf{Ei}}
\def\l{\left}
\def\r{\right}
\def\({\left(}
\def\){\right)}
\def\lb{\left\{}
\def\rb{\right\}}

\setcounter{page}{1}

\newcommand{\eref}[1]{(\ref{#1})}
\newcommand{\fig}[1]{Fig.\ \ref{#1}}

\def\bydef{:=}
\def\ba{{\mathbf{a}}}
\def\bb{{\mathbf{b}}}
\def\bc{{\mathbf{c}}}
\def\bd{{\mathbf{d}}}
\def\bee{{\mathbf{e}}}
\def\bff{{\mathbf{f}}}
\def\bg{{\mathbf{g}}}
\def\bh{{\mathbf{h}}}
\def\bi{{\mathbf{i}}}
\def\bj{{\mathbf{j}}}
\def\bk{{\mathbf{k}}}
\def\bl{{\mathbf{l}}}
\def\bm{{\mathbf{m}}}
\def\bn{{\mathbf{n}}}
\def\bo{{\mathbf{o}}}
\def\bp{{\mathbf{p}}}
\def\bq{{\mathbf{q}}}
\def\br{{\mathbf{r}}}
\def\bs{{\mathbf{s}}}
\def\bt{{\mathbf{t}}}
\def\bu{{\mathbf{u}}}
\def\bv{{\mathbf{v}}}
\def\bw{{\mathbf{w}}}
\def\bx{{\mathbf{x}}}
\def\by{{\mathbf{y}}}
\def\bz{{\mathbf{z}}}
\def\b0{{\mathbf{0}}}

\def\bA{{\mathbf{A}}}
\def\bB{{\mathbf{B}}}
\def\bC{{\mathbf{C}}}
\def\bD{{\mathbf{D}}}
\def\bE{{\mathbf{E}}}
\def\bF{{\mathbf{F}}}
\def\bG{{\mathbf{G}}}
\def\bH{{\mathbf{H}}}
\def\bI{{\mathbf{I}}}
\def\bJ{{\mathbf{J}}}
\def\bK{{\mathbf{K}}}
\def\bL{{\mathbf{L}}}
\def\bM{{\mathbf{M}}}
\def\bN{{\mathbf{N}}}
\def\bO{{\mathbf{O}}}
\def\bP{{\mathbf{P}}}
\def\bQ{{\mathbf{Q}}}
\def\bR{{\mathbf{R}}}
\def\bS{{\mathbf{S}}}
\def\bT{{\mathbf{T}}}
\def\bU{{\mathbf{U}}}
\def\bV{{\mathbf{V}}}
\def\bW{{\mathbf{W}}}
\def\bX{{\mathbf{X}}}
\def\bY{{\mathbf{Y}}}
\def\bZ{{\mathbf{Z}}}

\def\mA{{\mathbb{A}}}
\def\mB{{\mathbb{B}}}
\def\mC{{\mathbb{C}}}
\def\mD{{\mathbb{D}}}
\def\mE{{\mathbb{E}}}
\def\mF{{\mathbb{F}}}
\def\mG{{\mathbb{G}}}
\def\mH{{\mathbb{H}}}
\def\mI{{\mathbb{I}}}
\def\mJ{{\mathbb{J}}}
\def\mK{{\mathbb{K}}}
\def\mL{{\mathbb{L}}}
\def\mM{{\mathbb{M}}}
\def\mN{{\mathbb{N}}}
\def\mO{{\mathbb{O}}}
\def\mP{{\mathbb{P}}}
\def\mQ{{\mathbb{Q}}}
\def\mR{{\mathbb{R}}}
\def\mS{{\mathbb{S}}}
\def\mT{{\mathbb{T}}}
\def\mU{{\mathbb{U}}}
\def\mV{{\mathbb{V}}}
\def\mW{{\mathbb{W}}}
\def\mX{{\mathbb{X}}}
\def\mY{{\mathbb{Y}}}
\def\mZ{{\mathbb{Z}}}

\def\cA{\mathcal{A}}
\def\cB{\mathcal{B}}
\def\cC{\mathcal{C}}
\def\cD{\mathcal{D}}
\def\cE{\mathcal{E}}
\def\cF{\mathcal{F}}
\def\cG{\mathcal{G}}
\def\cH{\mathcal{H}}
\def\cI{\mathcal{I}}
\def\cJ{\mathcal{J}}
\def\cK{\mathcal{K}}
\def\cL{\mathcal{L}}
\def\cM{\mathcal{M}}
\def\cN{\mathcal{N}}
\def\cO{\mathcal{O}}
\def\cP{\mathcal{P}}
\def\cQ{\mathcal{Q}}
\def\cR{\mathcal{R}}
\def\cS{\mathcal{S}}
\def\cT{\mathcal{T}}
\def\cU{\mathcal{U}}
\def\cV{\mathcal{V}}
\def\cW{\mathcal{W}}
\def\cX{\mathcal{X}}
\def\cY{\mathcal{Y}}
\def\cZ{\mathcal{Z}}
\def\cd{\mathcal{d}}
\def\Mt{M_{t}}
\def\Mr{M_{r}}
\def\O{\Omega_{M_{t}}}
\newcommand{\figref}[1]{{Fig.}~\ref{#1}}
\newcommand{\tabref}[1]{{Table}~\ref{#1}}

\newcommand{\var}{\mathsf{var}}
\newcommand{\fb}{\tx{fb}}
\newcommand{\nf}{\tx{nf}}
\newcommand{\BC}{\tx{(bc)}}
\newcommand{\MAC}{\tx{(mac)}}
\newcommand{\Pout}{p_{\mathsf{out}}}
\newcommand{\nnn}{\nn\\}
\newcommand{\FB}{\tx{FB}}
\newcommand{\TX}{\tx{TX}}
\newcommand{\RX}{\tx{RX}}
\renewcommand{\mod}{\tx{mod}}
\newcommand{\m}[1]{\mathbf{#1}}
\newcommand{\td}[1]{\tilde{#1}}
\newcommand{\sbf}[1]{\scriptsize{\textbf{#1}}}
\newcommand{\stxt}[1]{\scriptsize{\textrm{#1}}}
\newcommand{\suml}[2]{\sum\limits_{#1}^{#2}}
\newcommand{\sumlk}{\sum\limits_{k=0}^{K-1}}
\newcommand{\eqhsp}{\hspace{10 pt}}
\newcommand{\tx}[1]{\texttt{#1}}
\newcommand{\Hz}{\ \tx{Hz}}
\newcommand{\sinc}{\tx{sinc}}
\newcommand{\tr}{\mathrm{tr}}
\newcommand{\diag}{\mathrm{diag}}
\newcommand{\MAI}{\tx{MAI}}
\newcommand{\ISI}{\tx{ISI}}
\newcommand{\IBI}{\tx{IBI}}
\newcommand{\CN}{\tx{CN}}
\newcommand{\CP}{\tx{CP}}
\newcommand{\ZP}{\tx{ZP}}
\newcommand{\ZF}{\tx{ZF}}
\newcommand{\SP}{\tx{SP}}
\newcommand{\MMSE}{\tx{MMSE}}
\newcommand{\MINF}{\tx{MINF}}
\newcommand{\RC}{\tx{MP}}
\newcommand{\MBER}{\tx{MBER}}
\newcommand{\MSNR}{\tx{MSNR}}
\newcommand{\MCAP}{\tx{MCAP}}
\newcommand{\vol}{\tx{vol}}
\newcommand{\ah}{\hat{g}}
\newcommand{\tg}{\tilde{g}}
\newcommand{\teta}{\tilde{\eta}}
\newcommand{\heta}{\hat{\eta}}
\newcommand{\uh}{\m{\hat{s}}}
\newcommand{\eh}{\m{\hat{\eta}}}
\newcommand{\hv}{\m{h}}
\newcommand{\hh}{\m{\hat{h}}}
\newcommand{\Po}{P_{\mathrm{out}}}
\newcommand{\Poh}{\hat{P}_{\mathrm{out}}}
\newcommand{\Ph}{\hat{\gamma}}
\newcommand{\mat}[1]{\begin{matrix}#1\end{matrix}}
\newcommand{\ud}{^{\dagger}}
\newcommand{\C}{\mathcal{C}}
\newcommand{\nn}{\nonumber}
\newcommand{\nInf}{U\rightarrow \infty}


\title{\huge Wireless Data Acquisition for Edge Learning: \\
Data-Importance Aware Retransmission }

\author{Dongzhu Liu, Guangxu Zhu, Jun Zhang, and Kaibin Huang
\thanks{\noindent  D. Liu, G. Zhu, and K. Huang are with the Dept. of Electrical and Electronic Engineering at The University of Hong Kong, Hong Kong. J. Zhang is with the Dept. of Electronic and Information Engineering at the Hong Kong Polytechnic University, Hong Kong. Corresponding author: K. Huang (email: huangkb@eee.hku.hk). 
}}
\maketitle

%
%
%

\vspace{-14mm}

\begin{abstract}
By deploying machine-learning algorithms at the network edge, edge learning can leverage the enormous real-time data generated by billions of mobile devices to train AI models, which enable intelligent mobile applications. In this emerging research area, one key direction is to efficiently utilize radio resources for wireless data acquisition to minimize the latency of executing a learning task at an edge server. Along this direction, we consider the specific problem of retransmission decision in each communication round to ensure both reliability and quantity of those training data  for accelerating model convergence. To solve the problem, a new retransmission protocol called \emph{data-importance aware automatic-repeat-request} (importance ARQ) is proposed. Unlike the classic ARQ focusing merely on reliability, importance ARQ selectively retransmits a data sample based on its \emph{uncertainty} which helps learning and can be measured using the model under training. Underpinning the proposed protocol is a derived elegant communication-learning relation between two corresponding metrics, i.e., \emph{signal-to-noise ratio} (SNR) and data uncertainty. This relation facilitates the design of a simple threshold based policy for importance ARQ. The policy is first derived based on the classic classifier model of \emph{support vector machine} (SVM), where the uncertainty of a data sample is measured by its distance to the decision boundary. The policy is then extended to the more complex model of \emph{convolutional neural networks} (CNN) where data uncertainty is measured by entropy. Extensive experiments have been conducted for both the SVM and CNN using real datasets with balanced and imbalanced distributions. Experimental results demonstrate that importance ARQ effectively copes with channel fading and noise in wireless data acquisition to achieve faster model convergence than the conventional channel-aware ARQ. The  gain is more significant when the dataset is imbalanced. 
\end{abstract}

\vspace{-3mm}
\section{Introduction} 
\vspace{-1mm}
With the prevalence of smartphones and \emph{Internet-of-Things} (IoT) sensors on the network edge, known as \emph{edge devices}, people envision an incoming new world of ubiquitous computing and ambient intelligence. This vision motivates Internet companies and telecommunication operators to develop technologies for deploying machine learning  on the (network) edge to support intelligent mobile applications, named as \emph{edge learning}  \cite{zhu2018towards,wang2018edge,mao2017survey,park2018wireless}. This trend aims at leveraging enormous real-time data generated by billions of edge devices to train AI models. In return, downloading the learnt intelligence onto the devices will enable them to respond to real-time events with human-like capabilities. Edge learning crosses two disciplines, wireless communication and machine learning, which cannot be decoupled as their performances are interwound under a common goal of fast learning. 

As data-processing speeds are increasing rapidly, wireless acquisition of high-dimensional training data from many edge devices has emerged to be a bottleneck for fast edge learning, which faces the challenges due to high mobility and unreliable devices (see e.g., \cite{bonawitz2019towards}). This calls for designing highly efficient techniques for radio resource management targeting edge learning. For conventional techniques, data bits (or symbols) are assumed of equal importance, which simplifies the design criterion to be rate maximization but fails to exploit the features of learning. In contrast, for learning, \emph{the importance distribution in a  training dataset is non-uniform}, namely that some samples are more important than others. For instance, for training a classifier, the samples near decision boundaries are more critical than those far away \cite{settles2012active}. This fact motivates the proposed design principle of \emph{importance-aware resource allocation}.  In this work, we apply this principle to redesign the classic technique of  \emph{automatic repeat-request} (ARQ) for efficient wireless data acquisition in edge learning. 

\vspace{-10pt}

\subsection{ Wireless Communications for Edge Learning}

Conventional communication techniques are designed mostly for either reliable transmission or data-rate maximization without awareness of data utility for learning. Such a ``communication-learning separation" principle does not yield efficient solutions for acquiring large-scale distributed data in edge learning. Its increasingly critical communication bottleneck calls for redesigning communication techniques with a new objective of low-latency execution of learning tasks. 
Research opportunities in this largely uncharted area can be roughly grouped under three topics: radio resource allocation, multiple access, and signal encoding. The new idea in radio resource allocation for edge learning, the topic of our interest, is to consider data usefulness for learning in allocating resources for data uploading from devices to a server \cite{chen2018label}. In this paper, we consider retransmission which is a simple time-allocation method for ensuring reliable communication in the presence of channel hostility  \cite{she2017radio}. The widely used ARQ protocols repeat the transmission of a data packet until it is reliably received. Thereby, channel uses are allocated to packets under a reliability constraint. While existing ARQ designs purely target data reliability \cite{onggosanusi2003hybrid,zhang1999hybrid}, accelerating edge learning calls for new protocols incorporating the new feature of considering data importance in retransmission decision. This motivates our work.

The second key topic in the area is low-latency multi-access for distributed edge learning. Recent research focuses on federated learning, where edge devices transmit their local model updates to collaboratively update the global AI model by aggregation at the server \cite{konevcny2016federated}. One idea proposed recently is to perform ``over-the-air" aggregation by exploiting the waveform superposition property of a multi-access channel \cite{zhu2018low,amiri2019machine,yang2018federated}. Such a scheme allows simultaneous access and hence can dramatically reduce multi-access latency. 

Last, signal encoding for communication efficient edge learning represents another research thrust. Relevant research aims at integrating feature extraction, source coding, and channel encoding to compress transmitted data without significantly compromising learning performance. Examples include analog encoding on Grassmannian for high mobility data classification \cite{du2018fast} and quantized stochastic gradient descent \cite{alistarh2017qsgd}.

\vspace{-10pt}

{\color{black}
\subsection{Wireless Data Acquisition}

Efficient data acquisition is a classic topic in designing \emph{wireless sensor network} (WSN) with a rich literature \cite{li2018wireless,guo2017massive,malak2016optimizing,zhan2018energy,nie2018quality}. The main challenge is how to overcome the energy constraints of sensors to allow fusion centers to collect distributed sensing data without interruptions.  There exist diversified solutions such as wireless power transfer \cite{li2018wireless}, multi-hop transmission \cite{guo2017massive,malak2016optimizing}, and UAV-assisted data collection \cite{zhan2018energy}. One approach that shares the same spirit as the current work is to schedule sensors based on their data quality evaluated using criteria including cost, sensing accuracy and timeliness \cite{nie2018quality}. On the other hand, the ARQ protocol proposed in the current work also involves data evaluation which, however, is based on a different criterion, namely importance for learning. Overall, data utilization (i.e., computing or learning) is considered out of scope in prior work and not accounted for in existing techniques for data acquisition, leaving some space for performance improvement.

In machine learning, one topic relevant to data acquisition is \emph{active learning} \cite{settles2012active}. Consider the scenario where unlabeled data are abundant but manually labeling is expensive.  Active learning aims to selectively label informative data  (by querying an oracle), such that a model can be trained using as few labelled data samples as possible, thus reducing the labelling cost. Roughly speaking, the informativeness of a sample is related to how uncertain the prediction of this sample is under the current model.  To be specific, the more uncertain the prediction is, the more useful the sample can be for model learning.  Several commonly used uncertainty measures are \emph{entropy} \cite{holub2008entropy}, \emph{expected model change}~\cite{settles2008multiple}, and \emph{expected error reduction}~\cite{roy2001toward}. In active learning, wireless communication is irrelevant.  However, the uncertainty measures developed therein are useful for this work and integrated with a retransmission protocol to enable intelligent data acquisition in an edge learning system. 
}

\vspace{-10 pt}

\subsection{Contributions and Organization} 

This work concerns wireless data acquisition in edge learning.   In this work, we propose a new retransmission protocol called \emph{data-importance aware ARQ}, or \emph{importance ARQ} for short, which adapts retransmission decisions to both data importance and reliability (or equivalently the channel state). As a result, the allocation of channel uses is biased towards protecting important data samples against channel noise while ensuring the quantity of acquired data. Balancing the two aspects in the design results in the combined effects of accelerating model convergence and reducing the required budget of channel uses. 
To the authors' best knowledge, this work represents the first attempt on exploiting the non-uniform distribution of data informativeness to improve the communication efficiency of an edge learning system. 

The main contributions of this work are summarized as follows. 

\begin{itemize}
\item{\bf Importance ARQ for SVM:} First, consider the classic classifier model of \emph{support vector machine} (SVM).  The importance ARQ is designed to improve the quality-vs-quantity tradeoff.  The protocol selectively retransmits a data sample based on its underlying importance for training an SVM model which is estimated using the real-time model under training. For SVM, a suitable importance measure is proposed to be the shortest distance from a data sample to decision boundaries. 
 The theoretical contribution of the design lies in a derived \emph{elegant communication-learning  relation} between two corresponding metrics, i.e., \emph{signal-to-noise ratio} (SNR) and  \emph{data importance}, for targeted learning performance. This new relation facilitates the design of a simple threshold based policy for making retransmission decisions, where the SNR threshold is shown to be proportional to the importance measure.

\item{\bf  Extension to general classifiers:} The derived importance-ARQ policy for  SVM models is extended to general classifier models. Particularly, the SNR threshold is designed to be proportional to a monotonically increasing \emph{reshaping function} of a general importance measure.  The design captures the heuristic that more important data should be better protected against noise by a higher target SNR. Moreover, general guidelines on how to select the reshaping function and the SNR-importance scaling factor are discussed.  Subsequently, a case study on designing importance ARQ for the modern \emph{convolutional neural networks} (CNN) classifier is presented. 

\item{\bf  Experiments:}  We evaluate the performance of the proposed importance ARQ via extensive experiments using real datasets with balanced and imbalanced distribution. The results demonstrate that the proposed method avoids learning performance degradation caused by channel fading and noise while achieving faster convergence than the conventional channel-aware ARQ. Furthermore, the performance gain is found to be more significant for the imbalanced data distribution.  
\end{itemize}

The remainder of the paper is organized as follows. Section~\ref{sec: system model} introduces the communication and learning models. Section~\ref{sec: motivations} presents some initial experimental results and motivates the design of an intelligent retransmission protocol. The principle of importance ARQ is proposed for SVM in Section~\ref{sec: Principle}. It is extended to general classifiers in Section~\ref{sec: extension}. Section~\ref{sec:simulation} provides experimental results, followed by concluding remarks in Section~\ref{sec: concluding remarks}.

%
%
%

\section{Communication and Learning Models }\label{sec: system model}

In this section, we first introduce the communication system model and learning models.  Then data uncertainty metrics are defined for different learning models.  

\vspace{-10pt}
\subsection{Communication System Model}
We consider an edge learning system as  shown in Fig.~\ref{Fig: system} comprising an edge server and multiple edge devices, each equipped with a single antenna. A machine learning classifier is trained at the server using a labelled dataset distributed over devices. Denote the $k$-th data sample $(\bx_k, c_k)$ with $\bx_k \in \mathbb{R}^p$,  $p$ its dimensions, and $c_k\in \{1, 2, \cdots, C\}$ its label. The devices time share the channel and take turn to transmit local data to the server.  The time sharing is coordinated by a channel-aware scheduler while importance-aware scheduling is noted to be an interesting direction for future investigation. 
Note that a label has a much smaller size than a data sample (e.g., a $0-9$ integer versus a vector of a million coefficients). Thus two separate channels are planned: a low-rate   \emph{label channel} and a high-rate  \emph{data channel}. The former is assumed to be noiseless for simplicity. Reliable uploading of data samples over the noisy and fading channel is  the bottleneck of  wireless data acquisition and  the focus of  this work.  Time is slotted into symbol durations, called \emph{slots}. Transmission of a data sample requires $p$ slots, called a \emph{symbol block}. 

Upon receiving a data sample, the edge server makes a  \emph{binary decision} on whether to request a retransmission to improve the sample quality or a new sample from the scheduled device.   The decision is communicated to the device by transmitting either a 
 \emph{positive  ACK} or a \emph{negative ACK}.  The device is assumed to have backlogged data. Upon receiving a request from the server, the device transmits 
 either the previous sample or a new sample randomly picked from its buffer.

The data channel is assumed to  follow block-fading, where the channel coefficient remains static within a symbol block  and is \emph{independent and identically distributed} (i.i.d.) over different  blocks.  The transmit data sample ${ \bf x}=\l[X_1,X_2,\cdots, X_p\r]^{\sf T}$ is a random vector.  During the $i$-th symbol block, the active device sends the data ${ \bf x}^{(i)}$ using linear analog modulation, yielding the received signal given by
\begin{equation}\label{channel_model}
{\bf y}^{(i)}=\sqrt{P} h^{(i)}{\bf x}^{(i)}+{\bf z}^{(i)},
\end{equation} 
where $P$ is the transmit  power,  the channel coefficient $h^{(i)}$ is a complex \emph{random variable} (r.v.) with a unit variance, and ${\bf z}^{(i)}$ 
is the \emph{additive white Gaussian noise} (AWGN) vector with the entries  following i.i.d. ${\cal CN}(0,\sigma^2)$ distributions. Analog uncoded transmission is assumed not only for tractability but also to allow fast data transmission \cite{marzetta2006fast} and a higher energy efficiency (compared with the digital counterpart) \cite{cui2005energy}.   We assume that perfect \emph{channel state information} (CSI) on $h^{(i)}$   is available at the  server. This allows the server to compute  the  instantaneous SNR of a received data sample and make the retransmission decision. 
\begin{figure}[t]
\begin{center}
{\includegraphics[width=16.5cm]{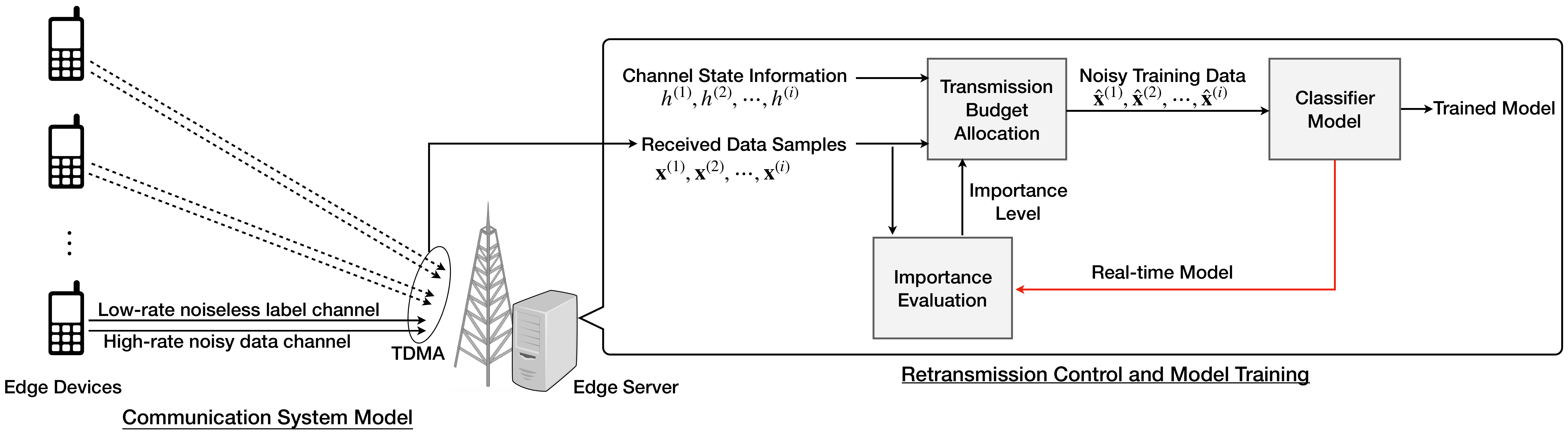}}\vspace{-6mm}
\caption{An edge learning system.}\vspace{-10mm}
\label{Fig: system}
\end{center}
\end{figure}

\subsubsection{Retransmission Combining} To exploit the time-diversity gain provided by multiple independent noisy observations of  the same data sample from retransmissions, the \emph{maximal-ratio combining} (MRC) technique is used to coherently combine all observations for maximizing  the receive SNR.  To be specific, consider  a data sample  $\bx$  retransmitted $T$ times.  All $T$ received copies, say from symbol block $n+1$ to $n+T$,  can be combined by MRC  to acquire the received sample, denoted as $\hat{\bx}(T)$, as follows:
\begin{equation} 
\hat{{ \bf x}}{(T)}=\frac{1}{\sqrt{P}}\Re\l(\sum_{i = n+1 }^{n+T}\frac{(h^{(i)})^*}{\sum_{m = n+1 }^{n+T}|h^{(m)}|^2}{\bf y}^{(i)}\r),\label{eq: est x}
\end{equation}
where ${\bf y}^{(i)}$ is given in \eqref{channel_model}. In \eqref{eq: est x}, we extract the real part of the combined signal for further processing since the data for learning  are real-valued in general (e.g., photos, voice clips or video clips).  As a result, the  effective receive SNR for $\hat{{ \bf x}}{(T)}$ after combining is given as 
\begin{equation}\label{eq: def SNR}
\SNR(T)=\frac{2P}{\sigma^2}\sum_{i = n+1 }^{n+T}|h^{(i)}|^2,
\end{equation}
where the coefficient $2$ at  the right-hand side arises from the fact that only the noise in the real dimension with variance  $\frac{\sigma^2}{2}$  affects the received data. The summation in \eqref{eq: est x} has a value growing as the number of retransmissions $T$ increases. The SNR expression in \eqref{eq: def SNR} measures the reliability of a received data sample  and  serves as a criterion for making the retransmission decision   as discussed  in Section~\ref{sec: Principle}. 

\subsubsection{Latency Constrained Transmission} Either due to the application-specific latency requirement  for the learning task or limited radio resources, the objective of designing the communication system is to minimize the duration of wireless data acquisition or equivalently maximize  the speed of model convergence. Under this objective,  the  retransmission protocol is designed in the sequel  to bias channel-use allocation towards providing better protection for more important data samples against channel noise.

\vspace{-8pt}

\subsection{Learning Models}
For the learning task, we consider supervised training of a classifier.  Prior to training, we assume that the edge server has a small set  of clean observed samples, denoted as ${\cal L}_0$.  This allows the construction of a coarse initial classifier, which is used for making  retransmission decisions at the beginning.  The classifier is refined  progressively in the data acquisition (and training) process.   
In this paper, we consider two widely used classifier models, i.e., the classic SVM classifier and the modern CNN classifier as introduced below.

\subsubsection{SVM Model} As shown in Fig.~\ref{Fig: SVM}, the SVM algorithm is to seek an optimal hyperplane ${\bf w} ^{\sf T}{\bf x}+b=0$ as a decision boundary by maximizing its margin $\gamma$ to data points, i.e., the minimum distance between the hyperplane to any data sample \cite{friedman2001elements}. 
%
 The points lie in the margin are referred to as \emph{support vectors} which directly determine the decision boundary. 
Let $({\bf x}_k, c_k)$ denote the $k$-th data-label pair in the training dataset. A convex optimization formulation for the SVM problem is given as
\begin{align}
\min_{{\bf w},b}\;\;  &\|{\bf w} \|^2 \label{eq:obj SVM}\\
{\rm s.t.} \;\;&c_k({\bf w} ^{\sf T}{\bf x}_k+b)\geq 1,  \quad   \forall k.
\end{align} 
The original SVM works only for linearly separable datasets, which is hardly the case when the dataset is corrupted by channel noise in the current scenario. To enable the algorithm to cope with a potential outlier caused by noise, a variant of SVM called \emph{soft margin SVM} is adopted. The technique is widely used in practice to classify a noisy dataset that is not linearly separable by allowing misclassification but with an additional penalty on the objective in \eqref{eq:obj SVM} (see \cite{friedman2001elements} for details).  After training, the learnt SVM model can be used for predicting the label of a new sample, denoted by ${\bf x}$, by computing its output score. The binary-classification case is as follows:
\begin{equation}\label{eq: def score}
\text{({\bf Output  Score}) }\quad s({\bf x})=({\bf w} ^{\sf T}{\bf x}+b)/\|\bw\|,
\end{equation}
where $\| \cdot\|$ represents the Euclidean norm and s({\bf x}) is a normalized score.  Then the sign of the output score yields the prediction of the binary label.

\begin{figure}[t]
\begin{center}
{\includegraphics[width=6cm]{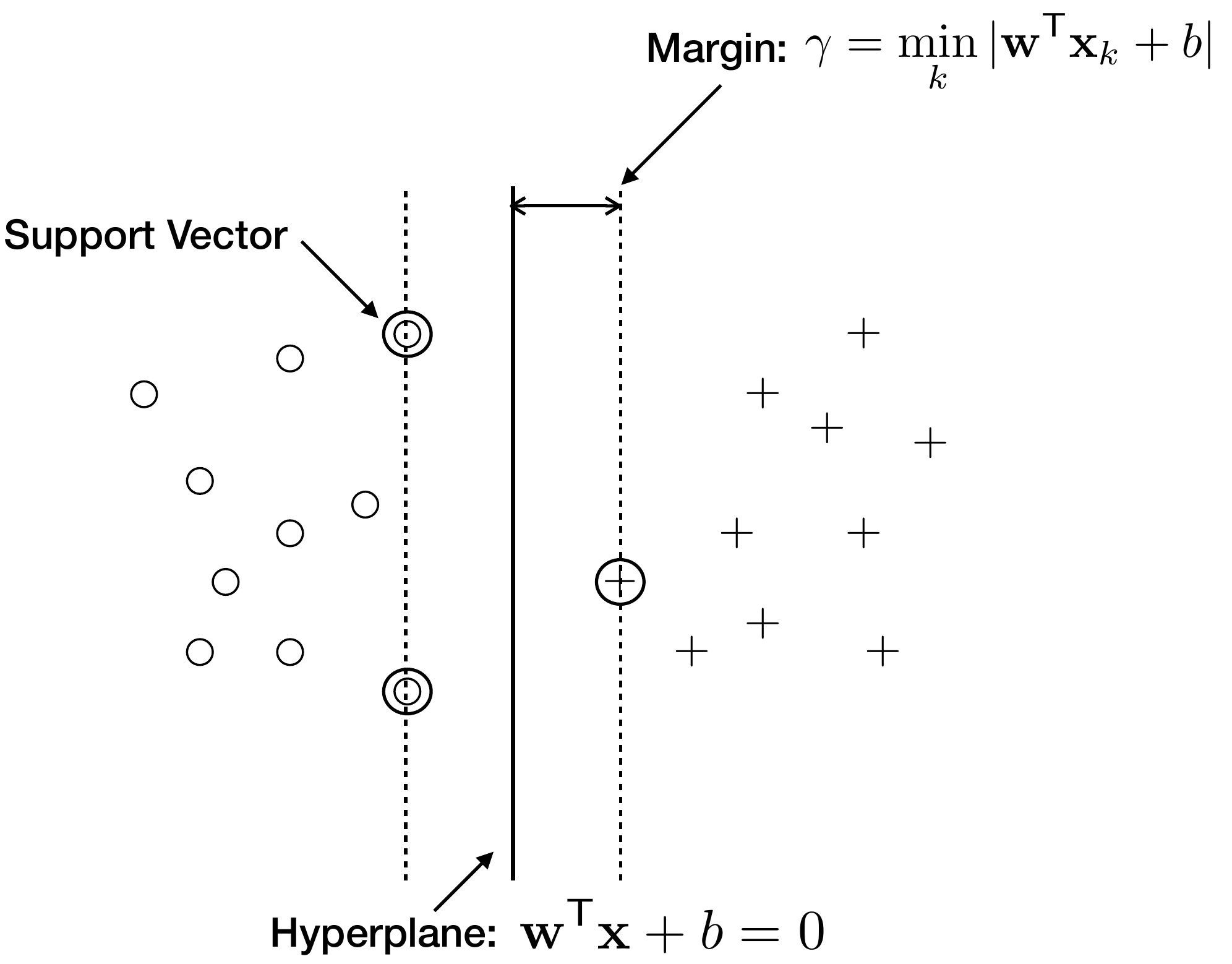}}\vspace{-6mm}
\caption{A binary SVM-classifier model.}\vspace{-6mm}
\label{Fig: SVM}
\vspace{-6mm}
\end{center}
\end{figure}

{\color{black}
\subsubsection{CNN model} CNN is made up of neurons that have adjustable weights and biases to express a non-linear mapping from an input data sample to class scores as outputs \cite{haykin1994neural}.  Fig.~\ref{Fig: CNN} illustrates the implementation of CNN, which consists of an input and an output layers, as well as multiple hidden layers. The hidden layers of a CNN typically include convolutional layers, ReLu layers, pooling layers, fully connected layers and normalization layers. 
Without the explicitly defined decision boundaries as for SVM, CNN adjusts the parameters of hidden layers to minimize the prediction error, calculated using the outputs of the softmax layer and the true labels of training data.  After training, the learnt CNN model can then be used for predicting the label of a new sample by choosing one with the highest posterior probability, which is obtained from the outputs of the softmax layer.  
}

\begin{figure}[t]
\begin{center}
{\includegraphics[width=15cm]{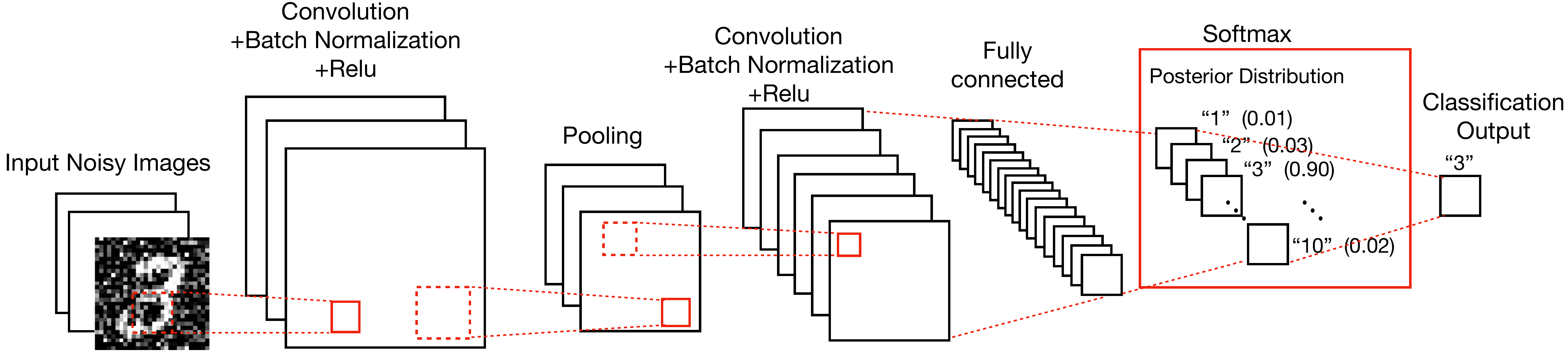}}\vspace{-15pt}
\caption{A CNN classifier model.}
\vspace{-10mm}
\label{Fig: CNN}
\end{center}
\end{figure}

\vspace{-10pt}
{\color{black}
\subsection{Data Uncertainty Metrics}\label{sec:Uncertainty}

The importance of a data sample for learning is usually measured by its \emph{uncertainty}, as viewed by the model under training \cite{settles2012active}.  Two uncertainty measures targeting SVM and CNN respectively are introduced as follows.



\subsubsection{Uncertainty Measure for SVM} For SVM, the uncertainty measure of a data sample is synonymous with its distance to the decision boundary \cite{tong2001support}. The definition is motivated by the fact that a classifier makes less confident inference on a  data sample which is located  near the decision boundary.  Based on this fact, we measure  the uncertainty of a data sample  by the inverse of its distance to the boundary.  Given a data sample ${\bf x}$ and a binary classifier, the said distance can be readily computed by the absolute value of the output score [see
\eqref{eq: def score}] as follows 
\begin{equation}\label{eq: def dist}
 d({\bf x}) =  | s({\bf x})|  = |{\bf w} ^{\sf T}{\bf x}+b|/ \|\bw\|.
\end{equation}
Then the distance based uncertainty measure is defined as
\begin{equation}\label{eq: uncertainty}
 \mathcal{U}_{\sf d}\l({\bf x}\r) = \frac{1}{ d^2({\bf x})} = {\|\bw\|^2}/{|{\bf w} ^{\sf T}{\bf x}+b|^2}.
\end{equation}
One can observe that the measure diverges as a data sample approaches the  decision boundary, and it reduces as the sample moves away from the boundary.

\subsubsection{Uncertainty Measure for CNN}  For CNN, a suitable measure is \emph{entropy}, an information theoretic notion, defined as follows \cite{holub2008entropy}:
%
\begin{equation}\label{eq: entropy}
 \mathcal{U}_{\sf e}\l({\bf x}\r)=-\sum_{c}P_{\theta}\l( c | \mathbf{ x}\r)\log P_{\theta}\l( c | \mathbf{ x}\r),
\end{equation}  
where $c$ denotes a class label and $\theta$ the set of model parameters to be learnt. To be precise, the model parameters are the weights and biases of the neurons in CNN.  

}

\section{Wireless Data Acquisition by Retransmission}\label{sec: motivations}

\subsection{Why Retransmission is Needed?}\label{sec: retransmission needed}

\begin{figure}[t]
\centering
\subfigure[Noiseless data.]{
\label{Fig: Mismatch_noiseless}
\includegraphics[width=5cm]{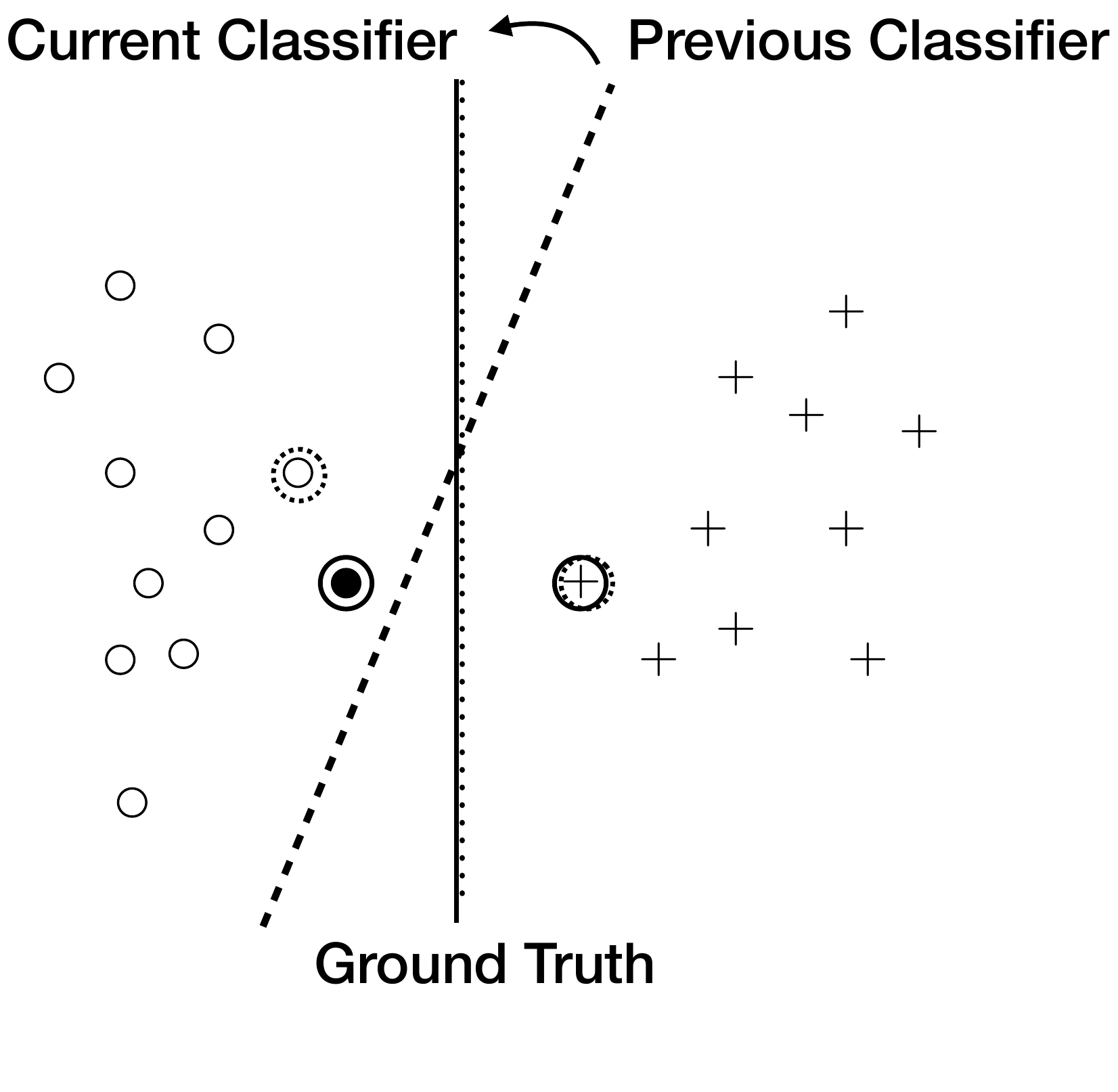}} 
\subfigure[Noisy data.]{
\label{Fig: Mismatch_noisy}
\includegraphics[width=6.7cm]{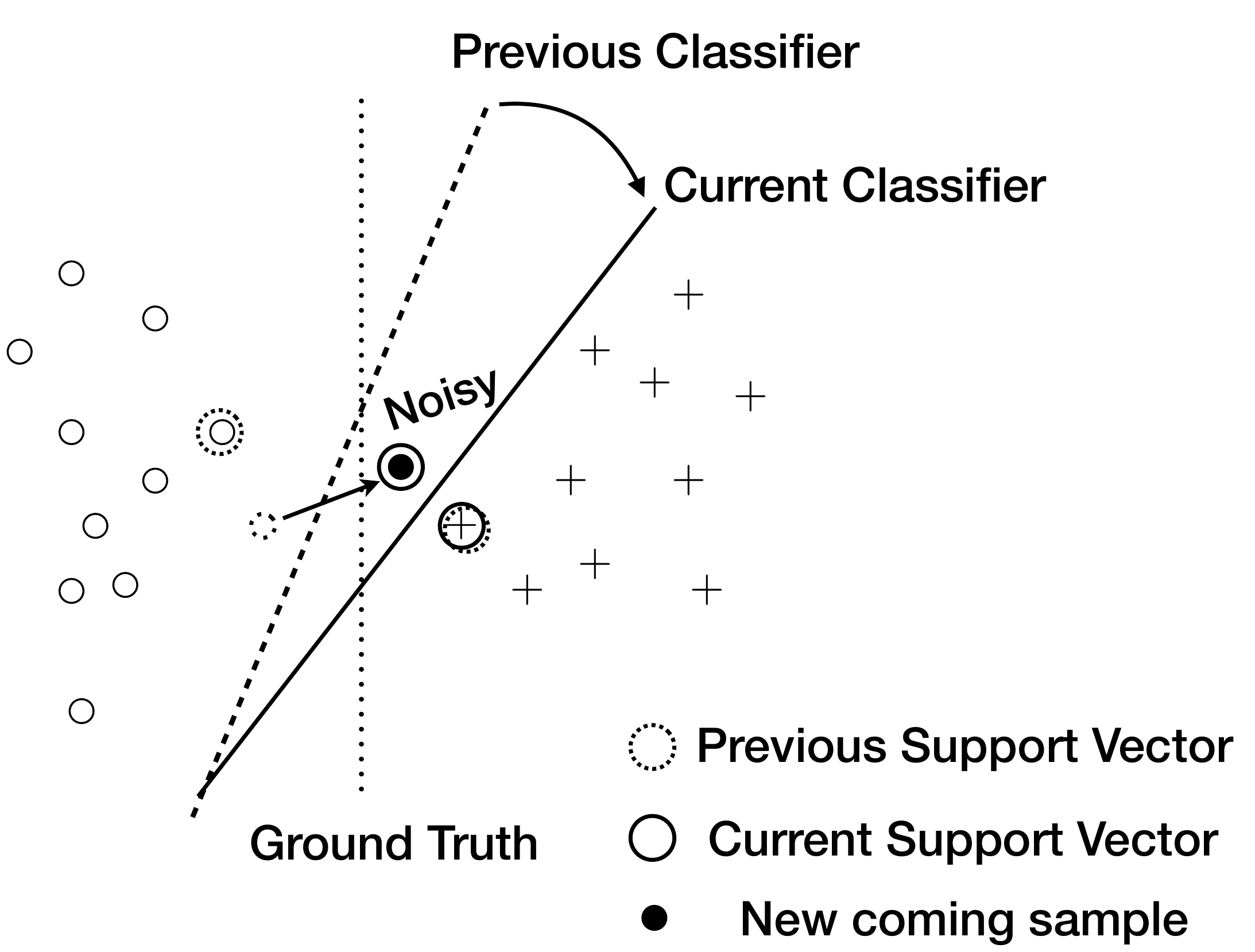}}
\caption{Illustration of the data-label mismatch issue for SVM.}
\vspace{-6mm}
\label{Fig: data label mismatch}
\end{figure}

Given a noisy data channel and a reliable label channel, the classifier model at the edge server is trained  using noisy data samples with  correct labels.  The channel noise and fading can cause a data sample to cross  the ground-truth decision boundary, thereby resulting a mismatch between the  sample and its  label, referred  to as a \emph{data-label mismatch}. The issue can cause incorrect learning as illustrated in Fig.~\ref{Fig: data label mismatch}. Specifically, for the noiseless transmission case in Fig.~\ref{Fig: Mismatch_noiseless}, the new data sample helps refine the current decision boundary to approach  the ground-truth one.  However, for the case of noisy transmission in Fig.~\ref{Fig: Mismatch_noisy}, noise  corrupts the  new sample and causes it to be an outlier falling into the opposite (wrong) side of the decision boundary.  The situation will be exacerbated when the SVM classifier is used since the outlier may be selected as the supporter vector (or indirectly affect the decision boundary by increasing the penalty in soft-margin SVM).

\begin{figure}[t]
\begin{center}
\subfigure[Learning performance for different number of retransmissions.]
{\includegraphics[width=10cm]{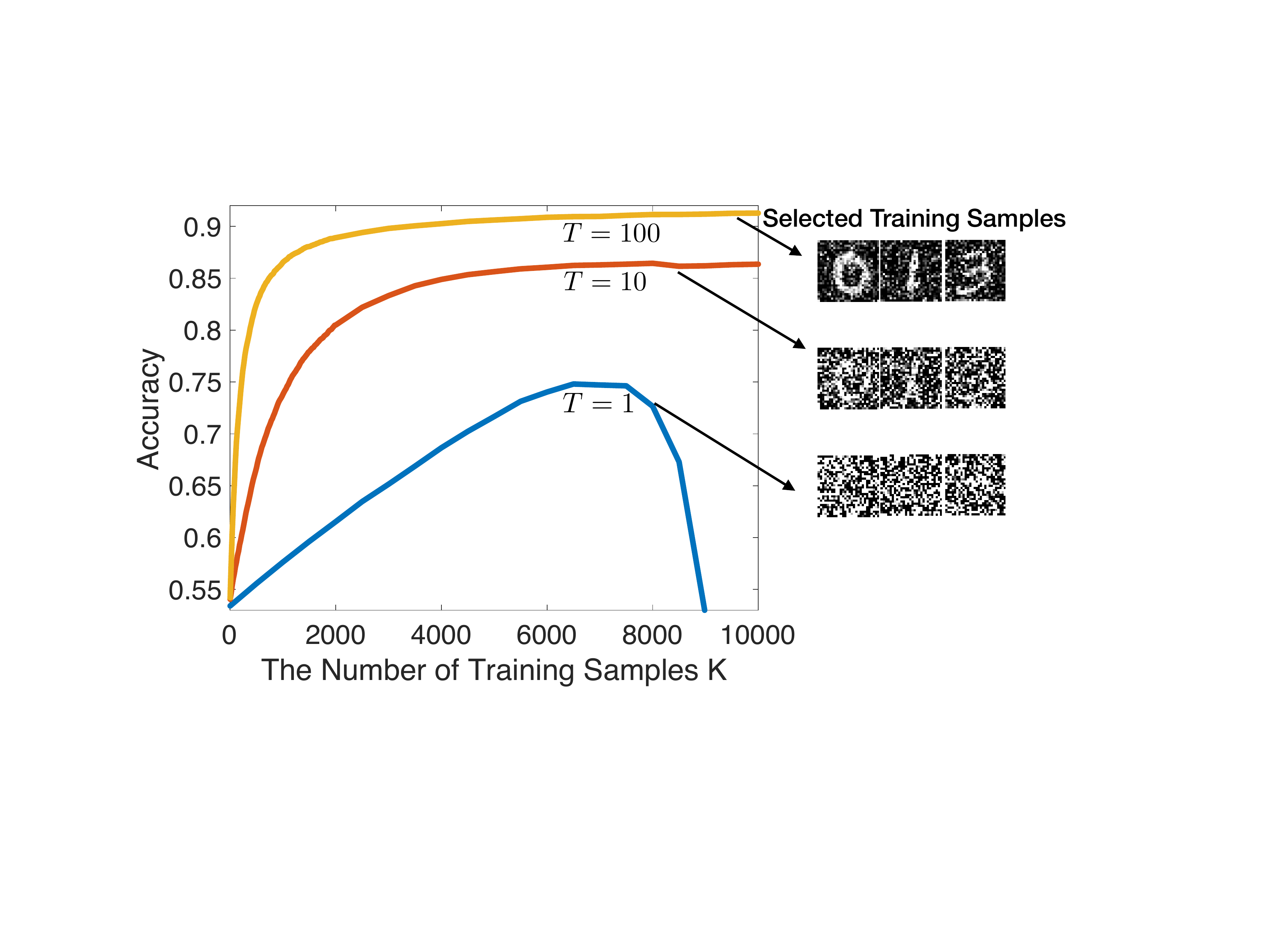}\label{fig: FixedRetrans}}\\
\subfigure[Visualization of received 10000 noisy training samples.]
{\includegraphics[width=12cm]{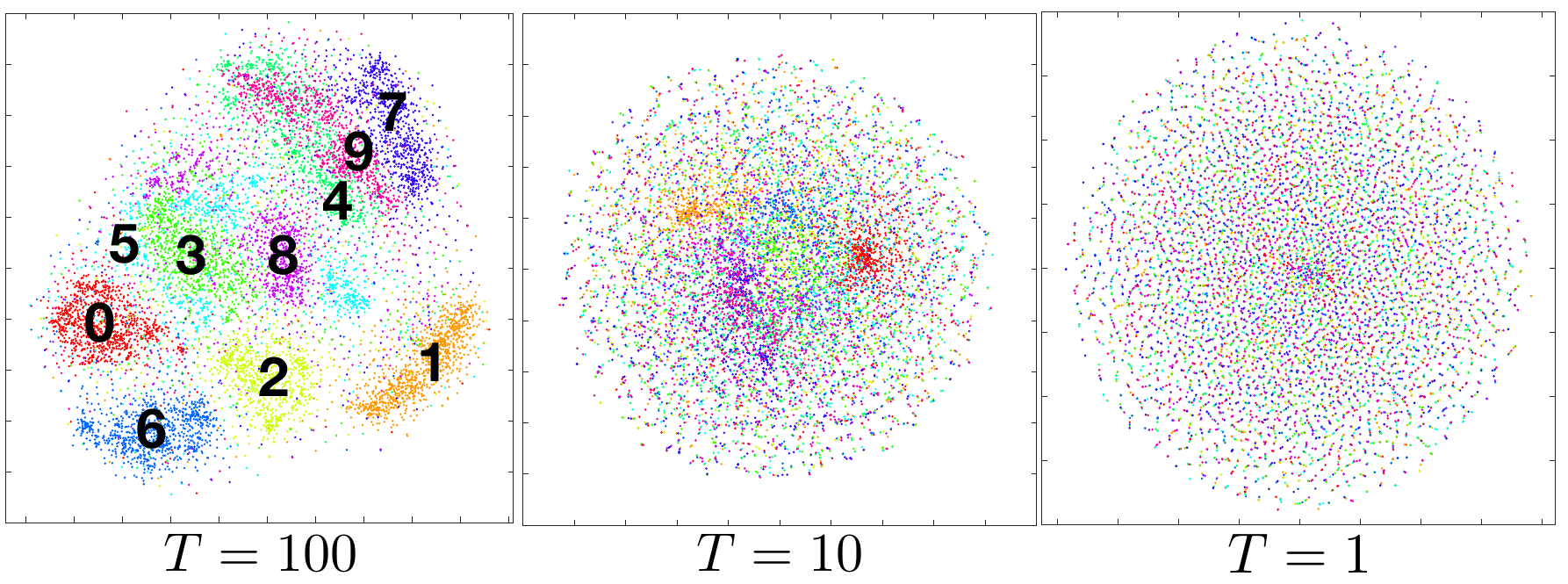}\label{fig: TSNE FixedRetrans}}
\caption{The impact of retransmission on the accuracy of the learnt model.}
\label{Fig: different SNR}
\vspace{-10mm}
\end{center}
\end{figure}

Retransmission can exploit time diversity to suppress channel noise and fading so as to improve data reliability and hence  the learning performance.  To visualize the benefit of retransmission, we compare in Fig. \ref{Fig: different SNR} the performance of classifiers which are trained using the noise corrupted  dataset with a varying number  of retransmissions. 
Specifically, we consider the learning  task of handwritten digit recognition using the well-known MNIST dataset that consists of 10 categories ranging from digit ``0'' to ``9'' \cite{lecun1998gradient}. The level of channel-noise is controlled by the average transmit SNR which is set as $\bar {\rho}=4 {\rm dB}$.  We train three SVM classifiers with different fixed numbers of retransmissions: $T=1, 10, 100$. The curves of their  test accuracy are shown  in  Fig.~\ref{fig: FixedRetrans}, with the corresponding received dataset  distribution visualized in Fig. \ref{fig: TSNE FixedRetrans}  using the classic 
\emph{$t$-distributed stochastic neighbor embedding} ($t$-SNE) algorithm for projecting the images onto the horizontal plane. It is observed from the case without retransmission ($T=1$), after receiving a certain number (i.e., 8000) of noisy data samples, 
the training of the classifier fails as reflected in the abrupt drop in test accuracy. The reason is that  the strong noise effect [see Fig. \ref{fig: FixedRetrans}]  accumulates to cause the divergence of the model [see Fig.~\ref{fig: TSNE FixedRetrans}]. As the number of retransmission increases, the noise effect is subdued to a sufficiently low level ensuring that the   class structure of the ideal dataset   can be resolved, leading to a converged model and a high  test accuracy. The experiment  demonstrates  the effectiveness of  retransmission in edge learning.  To further improve learning performance and more efficiently utilize the transmission budget, retransmission should be adapted to the importance levels of individual data samples, which is the focus of the remainder of the paper.  
%
\vspace{-12pt}

\subsection{Problem Statement}\label{sec: problem statement}

 The objective of designing importance ARQ is to adapt retransmission to both the data importance and the channel state so as to efficiently utilize  the finite transmission budget for optimizing the learning accuracy. The challenges faced by the design are reflected in two issues described as follows.  

\begin{itemize}
\item \emph{Quality-vs-Quantity Tradeoff}: The learning performance can be improved by either increasing the reliability (quality) of the wirelessly transmitted training dataset by more  retransmissions, or increasing its size (quantity) by acquiring more  data samples at the cost of their quality. Given a limited transmission budget, a tradeoff exists  between the two aspects, called the \emph{quality-vs-quantity tradeoff}. An efficient retransmission design must exploit the tradeoff to optimize the learning performance.  

\item \emph{Non-uniform Data Importance}: In conventional data communication, bits are implicitly assumed to have equal  importance. This is not the case for training a classifier where data samples with higher uncertainty are more informative and thus more important than those with lower uncertainty. Considering the non-uniform  importance in training data provides a new dimension for improving the communication efficiency, which should be also leveraged in the design.  
\end{itemize}


\vspace{-3mm}
\section{Data-Importance Aware Retransmission}\label{sec: Principle}
In this section, we consider the task of training an SVM classifier at the edge. First, the concept of noisy data alignment is introduced to relate wireless transmission and learning performance. By applying a relevant constraint, the importance ARQ protocol is derived  to intelligently allocate channel uses to the acquisition of individual data samples so as to accelerate model convergence. The protocol  is first designed for binary classification and then extended to multi-class classification.  

\vspace{-10pt}
\subsection{Probability of Noisy  Data Alignment}
The direct design of importance ARQ for optimizing the learning performance is difficult as there lacks a tractable mapping from data quality to learning accuracy. 
In this section, the difficulty is overcome by deriving a condition for retaining the usefulness of received data for learning in the presence of channel noise, which can differentiate data importance levels. The condition is derived based on the following fact: \emph{a  noisy received data sample can mislead the model training if its label  as predicted by the model differs from the ground truth received without noise}. To avoid this problem in the context of SVM,  a pair of  transmitted and received data samples should be  forced to lie at  the same side (ground-truth) of the decision hyperplane of the classifier model so that they have the same predicted labels. This event is  referred to as \emph{noisy data alignment} and denoted as ${\cal A}$.  Its probability  is called the \emph{data-alignment probability}.  From the distance based uncertainty defined in \eqref{eq: uncertainty} for SVM, one can see that data samples with higher uncertainty are more vulnerable  to noise corruption. To be specific,  a small noise perturbation  can push  a highly uncertain data sample  across  the decision boundary to result in the aforementioned data-label mismatch (see Fig. ~\ref{Fig: data label mismatch}). 
The high vulnerability of important data is the reason that importance ARQ allocates more resources to ensure their reliability, giving the protocol its name. The objective of designing importance ARQ is to satisfy  a constraint on the data-alignment probability. 

Next, the data-alignment probability is defined  mathematically for a binary classifier. Since the ground-truth model is unknown, the occurrence of the event ${\cal A}$ is evaluated using  the current  model under training  as a \emph{surrogate}. As a result, the  output scores  defined in \eqref{eq: def score} must yield the same signs for a pair of transmitted and received data samples if they are aligned.  Consider  an arbitrary transmitted data sample ${\bf x}$  and its received version $\hat{\bf x}(T)$  after $T$ retransmissions as defined in~\eqref{eq: est x}. The event ${\cal A}$ is specified as
\begin{equation}
\{{\cal A}\ |\ s({\bf x})s(\hat{\bf x}(T))>0\}.
\end{equation} 
Then  data alignment probability can be  mathematically  defined as follows.

\begin{definition}[Data-alignment probability]
\emph{Conditioned on the received data sample, the \emph{data-alignment probability} is defined as:
\begin{align}\label{conf_level}
{\cal P}\l(\hat{\bf x}(T)\r) = {\sf Pr} \l({\cal A} \; |\;  \hat{\bx}(T)\r).
\end{align} 
}
\end{definition}

The remainder of the sub-section is focused on analyzing the probability. To begin with, the distribution of the transmitted sample score $s(\bold{x})$ conditioned on the received data sample $\hat{\bold{x}}(T)$ can be obtained  from the conditional distribution of the transmitted sample, i.e., $p({\bf x} | \hat{\bf x}{(T)})$, as derived below.  
 
\begin{lemma}
\label{lemma: noiseless data distribution} \emph{Conditioned on the received  sample $\hat{\bf x}{(T)}$, the distribution of the transmitted sample $\bold{x}$ follows a Gaussian distribution:
\begin{equation}\label{eq:data distribution}
{\bf x} | \hat{\bf x}{(T)} \sim {\cal N}\l(\hat{\bf x}{(T)},\frac{1}{\SNR(T)} \bf I\r),
\end{equation}
where $\SNR(T)$ is the effective SNR given in \eqref{eq: def SNR}. 
 }
\end{lemma}
\proof {See Appendix~\ref{app: noiseless data distribution}.}

With the result,  the useful  distribution $p(s(\bold{x}) | \hat{\bf x}{(T)})$ can be readily derived using the linear relationship  in \eqref{eq: def score}.  The derivation simply involves projecting the high-dimensional Gaussian distribution onto a particular direction specified by $\bold{w}$,  which yields a univariate Gaussian distribution of dimension one as elaborated below.

\vspace{-10pt}

\begin{lemma}
\label{prop: Distribution of Distance to the Hyperplane} \emph{Conditioned on the estimated sample $\hat{\bf x}{(T)}$, the distribution of the transmitted sample score $s({\bf x})$  follows a unit-variate Gaussian distribution, given by 
\begin{align}\label{eq:distance distribution}
s({\bf x})|\hat{\bf x}{(T)}\sim {\cal N}\l(s(\hat{\bf x}{(T)}), \frac{1}{\SNR(T)}\r).
\end{align}
}
\end{lemma}

\begin{figure}[t]
\begin{center}
\subfigure[Small uncertainty.]
{\includegraphics[width=7cm]{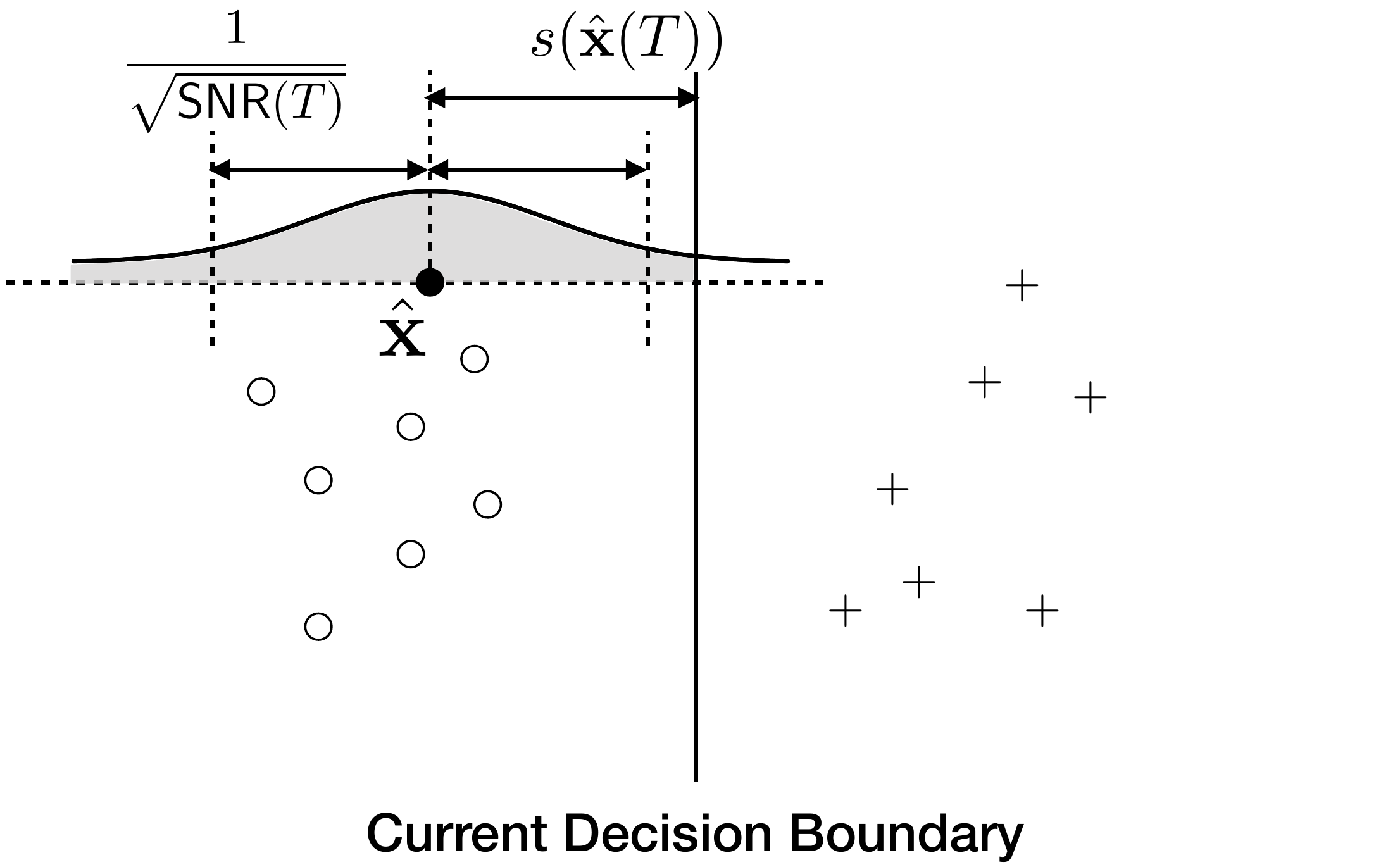}\label{Fig: confidence flat}} 
\subfigure[Lagre uncertainty.]
{\includegraphics[width=4.5cm]{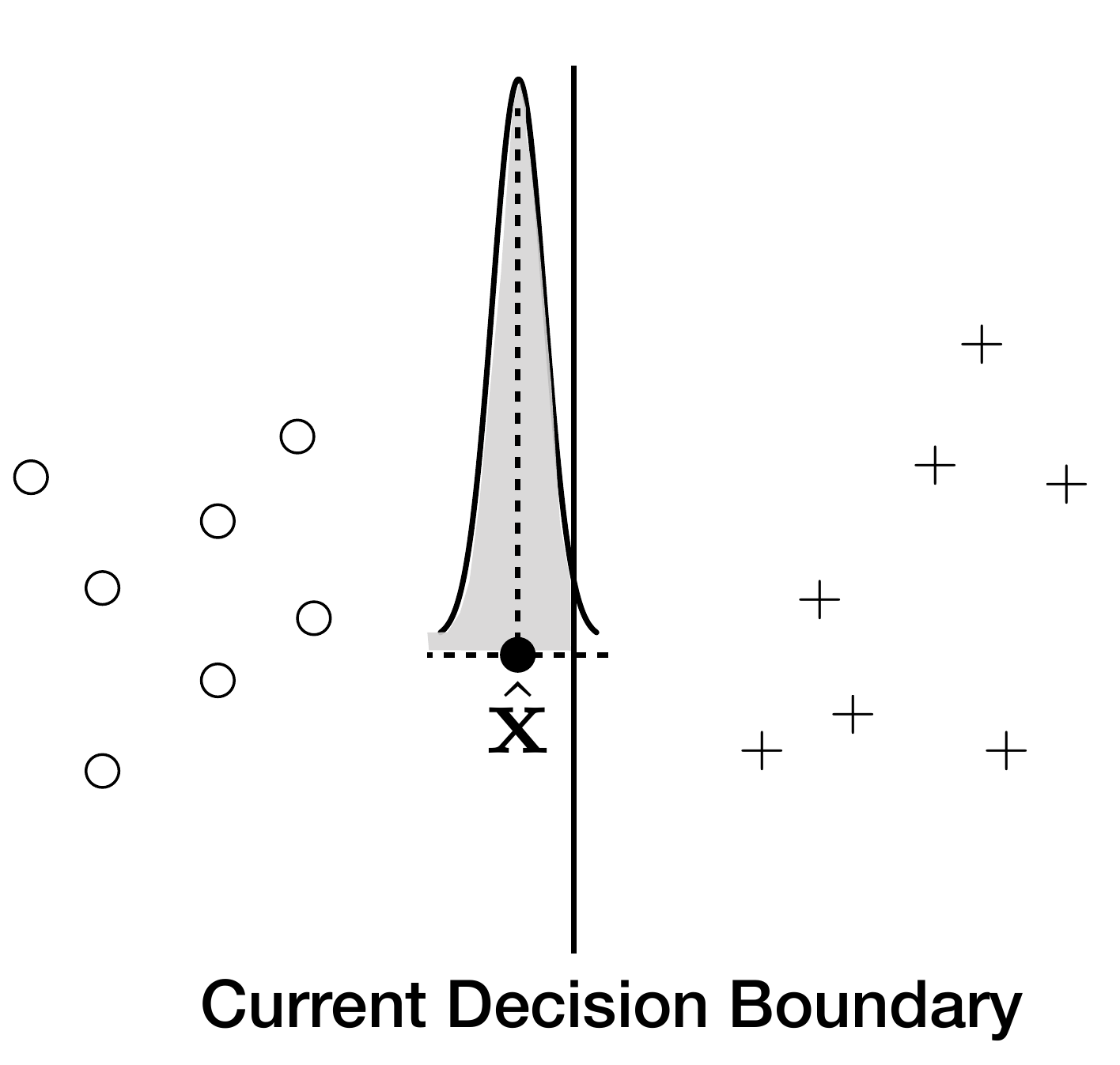}\label{Fig: confidence sharp}}
\caption{Illustration of the probability of noisy data alignment.}
\vspace{-10pt} 
\label{Fig: confidence level cal}
\vspace{-6mm}
\end{center}
\end{figure}

Based on Lemma \ref{prop: Distribution of Distance to the Hyperplane}, the data-alignment probability is presented in the following proposition.  
\vspace{-10pt}
\begin{proposition}
\label{prop: expr prob align} \emph{Consider the training of a  binary SVM classifier at the edge. Conditioned on the received sample $\hat{\bf x}{(T)}$, the data-alignment probability is given as 
\begin{align}\label{eq:derived align prob}
{\cal P}\l(\hat{\bf x}(T)\r) =\frac{1}{2}\l[1+{\rm erf }\l(\sqrt{\SNR(T)}\times\frac{|s(\hat{\bf x}(T))|}{ \sqrt{2}}\r)\r],
\end{align}
where $\mathrm{erf}(\cdot)$ is the well known  error function defined as $\mathrm{erf}(x)=\frac{2}{\sqrt{\pi}}\int_{0}^{x}e^{-t^2}dt$. }
\end{proposition}
\begin{proof}
As shown in Fig.~\ref{Fig: confidence level cal}, the conditional distribution for the transmitted data score $s({\bf x})$ is a Gaussian  and the probability of data alignment is equal to the area shaded in grey.   Mathematically, the probability can be derived using Lemma \ref{prop: Distribution of Distance to the Hyperplane} as follows:
\begin{align}
{\cal P}\l(\hat{\bf x}(T)\r) =&0.5+\sqrt{\frac{\SNR(T)}{2\pi }} \int_{0}^{|s(\hat{\bf x}(T))|} e^{-\SNR(T)\frac{ t^2}{2}}dt \label{eq: define sigma d}.
\end{align}
The integral therein can be further expressed using the error function $\mathrm{erf}(x)=\frac{2}{\sqrt{\pi}}\int_{0}^{x}e^{-t^2}dt$. 
\end{proof}


\begin{remark}\emph{(How does retransmission affects noisy data alignment?)\label{remark: retrans} \  Retransmission contributes to increasing  the data-alignment probability.  Specifically,  retransmission affects both the mean and variance of the conditional distribution $p(s({\bf x})|\hat{\bf x}{(T)})$ in \eqref{eq:distance distribution}. From the mean perspective, retransmission helps align  the average of retransmitted samples with  its ground truth. To be specific, the received estimate approaches the ground-truth value as the number of retransmissions grows:\vspace{-12pt}
\begin{align}
\vspace{-12pt}
\lim_{T \to \infty} s(\hat{\bf x}{(T)}) \to s({\bf x}).
\end{align}
From the variance perspective, retransmission continuously reduces the variance by increasing the receive SNR or equivalently the number of retransmissions $T$. Particularly, it follows from the definition of SNR [see \eqref{eq: def SNR}] that 
\begin{align}
\vspace{-8pt}
\frac{1}{\SNR(T)} = O({1}/{T})\quad \text{and}\quad \lim_{T \to \infty}  \frac{1}{\SNR(T)} \to 0.
\end{align}
Combining the two aspects, one can further apply the Chernoff bound to \eqref{eq: define sigma d} and obtain:
\vspace{-8pt}
\begin{align}
 {\cal P}\l(\hat{\bf x}(T)\r) = 1 - O(e^{- a T}),\vspace{-10pt}
\end{align}
where $a > 0$ is a positive constant.  
As a result, the probability of noisy data alignment approaches one at an exponential rate as $T$ grows.
}
\end{remark}
\vspace{-8pt}

Last, given the data alignment probability in \eqref{eq:derived align prob}, it is ready to specify the aforementioned condition for ensuring the usefulness of wirelessly acquired data for learning as the following constraint on a received sample $\hat{\bf x}(T)$ with $T$ retransmissions: 
\begin{equation}\label{eq: confident level prob}
\vspace{-8pt}
\text{({\bf Data Alignment Constraint})}\quad {\cal P}\l(\hat{\bf x}(T) \r)>p_c,
\end{equation}
where $p_c \in (0.5, 1)$ is a given constant. 

\subsection{Importance ARQ for Binary Classification}\label{subsec: imp ARQ binary}
In this section, the importance ARQ protocol is designed for binary SVM classification under the data alignment constraint in \eqref{eq: confident level prob} and the optimal control policy is proved to have a threshold based structure. 

First, it is shown  that the constraint in \eqref{eq: confident level prob} leads to a varying receive-SNR constraint on a data sample that depends on its importance level. The  result is given below, which follows directly from the monotonicity of the error function.

\begin{proposition}\label{prop: adaptive threshold}
\label{lemma: confidence level requirement}\emph{Consider the training of a binary SVM classifier at the edge. For a received data sample $\hat{\bf x}(T)$, the data alignment constraint in \eqref{eq: confident level prob} is satisfied if and only if the receive SNR exceeds an  importance based threshold:
\begin{equation}\label{eq:confidence level requirement}
{ \SNR}{(T)}>\theta_{0}  \; \mathcal{U}_{\sf d}\l(\hat{\bf x}(T)\r), \vspace{-8pt}
\end{equation}
where $\mathcal{U}_{\sf d}\l(\cdot\r)$ is the uncertainty measure given in \eqref{eq: uncertainty} and $\theta_{0}= \l[\sqrt{2} {\rm erf}^{-1}\l(2p_c-1\r)\r]^2$.  
 }
\end{proposition}

It is remarked that the scaling factor $\theta_{0}$ in \eqref{eq:confidence level requirement} can be interpreted as a conversion ratio specifying the rate at which the uncertainty measure is translated into the SNR requirement. The factor grows as the data-alignment constraint,  $p_c$, becomes more stringent, and vice versa. 

Next, using the result in Proposition~\ref{prop: adaptive threshold}, the importance ARQ protocol is designed as follows. Since the effective receive SNR after combining is a monotone increasing function of the number of retransmission, the constraint in \eqref{eq: confident level prob} can be translated into a threshold based retransmission policy. On the other hand, the SNR threshold in \eqref{eq:confidence level requirement} can diverge for an extremely uncertain data sample. Hence, it is necessary to limit the threshold value to avoid resource-wasteful excessive retransmission. The resultant simple protocol is described as follows. 

\begin{framed}
\vspace{-10pt} 
\begin{protocol}[Importance ARQ for binary SVM classification]\label{scheme: importance ARQ}\emph{Consider the acquisition of a  data sample ${\bf x}$ from a scheduled edge device. The edge server repeatedly requests the device to retransmit   ${\bf x}$ until the effective receive SNR  satisfies 
\vspace{-8pt}
\begin{equation}\label{eq: adaptive SNR form}
{\!\! \SNR}{(T)}\!>\! \min ( \theta_{0} \; \mathcal{U}_{\sf d}\l(\hat{\bf x}(T)\r), \theta_{\SNR}),\vspace{-8pt}
\end{equation} 
where $\theta_{\sf SNR}$ is a given maximum  SNR.
}
\end{protocol}
\vspace{-6pt} 
\end{framed}

\begin{remark}[Importance-aware SNR control]\label{remark: adp SNR to un}\emph{The importance ARQ protocol is a threshold based control policy with a  SNR threshold adapted to data importance. From \eqref{eq: adaptive SNR form}, the SNR threshold is proportional to the distance-based uncertainty of the data sample, $\mathcal{U}_{\sf d}\l({\bf x}\r)$.  It is aligned with the intuition  that  a data sample of higher uncertainty should be more reliably  received.  To better understand this result, a graphical illustration is provided in Fig.~\ref{Fig: confidence level cal}.  For a pre-specified $p_c$, a highly uncertain sample near the decision hyperplane requires a slim  distribution  with small variance (corresponding to a higher receive SNR and hence more retransmissions) to meet the requirement on  the data-alignment probability (the area shaded in grey) to be larger than $p_c$ [see Fig.~\ref{Fig: confidence sharp}]. On the other hand,  for  a less uncertain data sample,  the requirement of $p_c$ can be easily satisfied  with a relatively flat distribution with a large variance and low receive SNR [see Fig.~\ref{Fig: confidence flat}].  
}
\end{remark}

Last, the importance ARQ protocol is compared with the conventional channel-aware counterpart. For the latter, the retransmission policy is merely  channel-aware, and a fixed SNR threshold is set for all data samples without differentiating their importance, as described below. 

\vspace{-10pt} 
\begin{framed}
\vspace{-15pt} 
\begin{protocol}[Channel-aware ARQ]\emph{Consider the acquisition of a  data sample ${\bf x}$ from a scheduled edge device. The edge server repeatedly requests the device to retransmit  ${\bf x}$ until the required  effective SNR, $\theta_{\sf SNR}$,   is attained: 
\begin{equation}\label{eq:channel aware}
{ \SNR}{(T)}>\theta_{\sf SNR},
\end{equation}
where ${ \SNR}{(T)}$ is defined in \eqref{eq: def SNR}.}
\end{protocol}
\vspace{-15pt} 
\end{framed}
\vspace{-15pt} 
\begin{remark}[Uniform vs. heterogenous reliability] \emph{As the SNR requirement in \eqref{eq:channel aware} is independent of data uncertainty, the channel-aware  protocol achieves \emph{uniform reliability} for data samples. If deployed in an edge learning system, it can lead to inefficient utilization of radio resource due to unnecessary retransmissions for unimportant data, resulting in sub-optimal learning performance. In contrast, the proposed importance ARQ protocol achieves \emph{heterogeneous reliability} for data samples according to their  importance levels. This allows more efficient resource utilization via improving the quality-vs-quantity tradeoff, thereby accelerating learning.}
\end{remark} 

\vspace{-15pt}
\subsection{Implementation of  Multi-Class Classification}\label{sec: multi-class}
 In this subsection, the principle of importance ARQ developed in the preceding sub-section for binary classification  is  generalized to  multi-class classification.  Note that a $C$-class SVM classifier can be trained using the  so-called \emph{one-versus-one} implementation \cite{platt2000large}. The implementation decomposes the classifier into  $L = C(C-1)/2$ \emph{binary component classifiers} each trained using the samples from the two concerned classes only. As a result, for each input data sample $\bold{x}$, a $C$-class SVM outputs a $L$-dimension vector, denoted as  $\bold{s} = [s_1({\bf x}),s_2({\bf x}), \cdots, s_L({\bf x})]$, which records the $L$ output scores as  defined in \eqref{eq: def score}, from the component classifiers. To map the output $\bold{s}$ to one of the class indexes, a so-called \emph{reference coding matrix} of size $C \times L$ is built and denoted by $\bold{M}$, where each row gives the ``reference output pattern'' corresponding to the associated class. An example of the reference coding matrix with $C$ = 4  and hence $6$ binary component classifiers   is provided as follows:
\vspace{-4mm}
\begin{small}
\begin{align*}& \bold{M}  \!\!= \quad
\bordermatrix{%
    & _{\rm  binary 1}      &  _{\rm binary 2}    & _{\rm binary 3}   &_ {\rm binary 4} & _{\rm binary 5} & _{\rm binary 6} 
\cr 
_{{\rm class1}}    & 1         & 1       &1     & 0  & 0  & 0\cr
_{\rm class2}    & -1        & 0       &0     & 1  &1   &0 \cr
_{\rm class3}    &0          &-1      &0      &-1  &0    &1 \cr
_{\rm class4}    & 0         & 0       &-1     &0   &-1   &-1 \cr
}.
\end{align*}
\end{small}
\vspace{-4mm}

\noindent Given  $\bold{M}$, the prediction of the class index of $\bold{s}$  involves simply comparing the Hamming distances between $\bold{s}$ and different rows in $\bold{M}$, and choosing the row index with the smallest distance as the predicted class index. Particularly, the Hamming distance between $\bold{s}$ and the $c$-th row of $\bold{M}$ is defined by
\vspace{-10pt}
\begin{align}\label{Hamming_dist}
d(\bold{s}, \bold{m}_c) = \sum_{\ell=1}^{L}|m_{c\ell}|[1-{\rm sgn}(m_{c\ell}s_\ell({\bf x}))]/2,
\end{align}
 where $m_{c\ell}$ denotes the $\ell$-th element in vector $\bold{m}_c$, and $\rm sgn(x)$ denotes the sign function taking a value from $\{1,0,-1\}$ corresponding to the cases $x>0$, $x=0$ and $x<0$, respectively.  
 One can observe from the distance definition that not all the component classifiers' output scores have an effect on predicting a particular class. For example, the scores from binary classifiers  $2,3$ and $6$ have no effect on determining class $2$ as they are assigned a zero weight in computing the Hamming distance between $\bold{s}$ and $\bold{m}_2$. In other words, only binary classifiers  $1,4$ and $5$ are active when class $2$ is predicted.
 
Having obtained the predicted label $\hat{c} = \arg\min_{c} d(\bold{s}, \bold{m}_c)$, all the active component classifiers determining the current predicted label should satisfy the requirement of data alignment probability predefined in \eqref{eq: confident level prob}.  Consequently, the \emph{single-threshold policy} for importance ARQ defined in \eqref{eq: adaptive SNR form} can be then extended to a \emph{multi-threshold policy} as defined below:
 \begin{equation}\label{eq: adaptive SNR form multicalss}
{ \SNR}{(T)}>\frac{\theta_{0}}{|s_\ell(\hat{\bf x}(T))|^2}, \quad \forall  \ell \in \{\ell \; |\; m_{\hat{c}\ell} \neq 0\}.
\end{equation}

\section{Extension to General Classifiers}\label{sec: extension}
 
In this section, we extend the proposed importance ARQ protocol designed in the preceding section for the SVM classifier  model  to a general   model, and present a case study using the modern CNN model.   

\vspace{-10pt}
\subsection{Importance ARQ for a Generic Model}
The derivation of  Protocol \ref{scheme: importance ARQ} targets for SVM and may not be directly extended to a generic classifier model (e.g., CNN), due to the lack of explicitly defined decision boundaries, and thus an explicit distance based uncertain measure. Nevertheless, the following   insight derived for the SVM model  is  applicable to a generic model: \emph{the receive-SNR threshold in wireless data acquisition with retransmission should be adapted to data uncertainty.} This motivates the generalization of the importance ARQ protocol by modifying  Protocol $1$ as follows. 

\vspace{-15pt} 
\begin{framed}
\vspace{-10pt} 
\begin{protocol}[Importance ARQ for generic classifier]\label{scheme: importance ARQ general}\emph{Consider the acquisition of a  data sample ${\bf x}$ from a scheduled edge device. The edge server repeatedly requests the device to retransmit  ${\bf x}$ until 
\begin{equation}\label{eq: adaptive SNR form general}
{\!\! \SNR}{(T)}\!>\! \min \Big(\theta_{0}\; \mathcal{L}(\mathcal{U}_{\sf x}\l(\hat{\bf x}(T)\r)) , \theta_{\SNR}\Big),
\end{equation}
where $\mathcal{U}_{\sf x}$  is an uncertainty measure,  $\theta_{0}$ is a given conversion  ratio between the uncertainty measure and the target SNR,   and $\mathcal{L}(\cdot)$ is a monotonically increasing function. 
\vspace{-15pt} }
\end{protocol}
\end{framed}
\vspace{-10pt} 

The main difference of the generic protocol from Protocol $1$ for SVM is that the distance-based uncertainty measure  in the latter is replaced by a general \emph{monotonically increasing} function of a  general uncertainty measure. The function is called (uncertainty)  \emph{reshaping function}. The  main motivation for introducing the function  is to accommodate various forms of uncertainty measures. In particular, this function provides the flexibility to reshape a  selected uncertainty  measure to allow it to have certain desired properties  as discussed in the sequel. Furthermore, the monotonicity of the function enforces the intuition that more uncertain data should be more reliably received.  

To apply the general Protocol~\ref{scheme: importance ARQ general} to training a specific classifier  model, the uncertainty measure, the reshaping function, and  the conversion ratio should be carefully designed  for efficient radio-resource utilization to achieve the desired  learning performance. Several {\bf design guidelines} are provided as follows.
\begin{itemize}
\item \emph{Selection of Uncertainty Measure:}
In general, the uncertainty measure should be selected for ease of computation according to the output of the learning model.  For example, for  SVM, the output score evaluated by \emph{linear decision boundaries} allows easy evaluation of the distance-based uncertainty in \eqref{eq: def dist}. In contrast, for  CNN, the \emph{softmax} output, which gives the posterior probability for each predicted class, makes the \emph{entropy} in \eqref{eq: entropy} a more natural choice for measuring  uncertainty.  

\item\emph{Design of Reshaping Function and Conversion Ratio:}
The reshaping function and the conversion ratio  should be jointly designed to address the following two practical issues. 

\begin{itemize}
\item \emph{Unregulated SNR for Data with Zero Uncertainty:} The minimum value of some uncertainty measures, e.g. entropy,  can be zero. Its direct use  in  \eqref{eq: adaptive SNR form general} without proper modification may lead to a corrupted training dataset. To be more specific, since the corresponding SNR thresholds have zero values,  data samples with zero uncertainty fail to  trigger retransmission and thus may be  received with unacceptably low  reliability in the case of strong noise. The use of such  corrupted  data in model training  can  cause model  divergence. This issue can be addressed  by a proper design of the reshaping function. 

\item \emph{Low Differentiability in SNR Threshold:} An issue can arise in practice due to a  narrow dynamic range of a selected uncertainty measure.  For example, if the uncertainty is measured by entropy, the corresponding dynamic range is given by $ \mathcal{U}_{\sf e}\l({\bf x}\r) \in [0,\log C]$, where $C$ denotes the number of classes. For $10$-class classification, we have $\mathcal{U}_{\sf e}\l({\bf x}\r)\in [0, 2.3]$, which can be too narrow in retransmission implementation.    In particular, without any reshaping function or a suitable conversion ratio, the SNR thresholds set as in  \eqref{eq: adaptive SNR form general} for the most and least important data would be about  the same, making importance ARQ insensitive to uncertainty and barely ``importance aware''.  
\end{itemize}

\end{itemize}

\vspace{-13pt}
\subsection{Implementing Importance ARQ for CNN}

In this subsection, we use CNN as an example to illustrate how the generic  importance ARQ in Protocol $3$  can be particularized to a mode of choice  based on the guidelines in the preceding sub-section. To begin with, as discussed,  \emph{entropy} is chosen as a suitable measure of  data uncertainty for  CNN.
Then, we design the reshaping function to have  the following form: $ \mathcal{L}(x) = 1+\gamma x$, where $\gamma$ is a scaling factor to be determined in the sequel.\footnote{An alternative such as  the nonlinear increasing functions $\mathcal{L}(x) = {(1+x)}^\gamma$ is also a suitable  choice as verified by experiments.}  Note that the bias term $1$ in $\mathcal{L}(x)$ is added to address the issue  of  zero SNR threshold. Particularly, we set the bias term to be $1$ rather than other positive values as it allows the conversion ratio   $\theta_{0}$ to be also interpreted as the minimum quality requirement for the least uncertain data with the entropy being zero.  This allows $\theta_{0}$ to be set easily following the typical settings in a wireless communication system (e.g.,  $\theta_{0}$ = $10$ dB). Note from \eqref{eq: adaptive SNR form general} that  $\theta_{\SNR}$ denotes the maximum quality requirement for the data with  the largest  uncertainty. Thus the scaling factor $\gamma$ can be determined by solving the equality
$\theta_{0}\l[1+\gamma \mathcal{U}_{\max}\r]=\theta_{\SNR}$ where the maximum entropy  $\mathcal{U}_{\max} = \log C$.  The above designs lead to the importance ARQ for the CNN classifier as shown below.
\vspace{-15pt} 
\begin{framed}
\vspace{-10pt} 
\begin{protocol}[Importance ARQ for CNN]\label{scheme: importance ARQ CNN}\emph{Consider the acquisition of a  data sample ${\bf x}$ from a scheduled edge device for training a CNN classifier model.   The edge server repeatedly requests the device to retransmit   ${\bf x}$ until 
\begin{equation}\label{eq: adaptive SNR form CNN}
{\!\! \SNR}{(T)}\!>\! \min \Big(\theta_{0}\l[1+\gamma \; \mathcal{U}_{\sf e}\l(\hat{\bf x}(T)\r)\r] , \theta_{\SNR}\Big),
\end{equation}
where $\gamma$ is a scaling factor given as $\gamma=\frac{1}{\mathcal{U}_{\max}}\l(\frac{\theta_{\SNR}}{\theta_{0}}-1\r)$. }
\vspace{-15pt} 
\end{protocol}
\end{framed}
\vspace{-10pt} 

\section{Experimental Results}\label{sec:simulation}

\vspace{-10pt}

\subsection{Experiment Setup}
\subsubsection{Channel Model} We assume the classic Rayleigh fading channel with channel coefficients following i.i.d. complex Gaussian distribution ${\cal CN}(0,1)$. The average transmit SNR defined as {\color{black}$\bar{\rho}=P/\sigma^2$ is by default set as $4$ dB.

\subsubsection{Learning Performance Metrics} 
The performance metrics are defined separately for the   cases of  \emph{balanced} and \emph{imbalanced} training datasets, depending on whether the dataset has more instances of certain classes than others.  A balanced dataset is an ideal setting while the  imbalanced setting is more likely to happen in real-world applications, e.g., fraud detection, medical diagnosis and network intrusion detection \cite{sun2009classification}. Given a balanced dataset, the learning performance is measured by the test accuracy. However, the overall accuracy is unable to reflect the performance using a  highly skewed dataset. For example, a naive  classifier that predicts all test samples as the majority class could achieve a high accuracy. However, it is unable to detect the minority  but critical class. To tackle the issue, two performance  metrics, i.e.,  G-mean and F-measure, widely used for imbalanced classification are adopted \cite{sun2009classification}. Both  are based on the following confusion matrix defined for binary classification for imbalanced data, where positive and negative classes correspond to minority and majority classes, respectively: 
\begin{small}
\begin{align*}& \bold{Confusion \ Matrix} \! = \
\bordermatrix{%
    & {\rm  predicted\ positive}      &  {\rm predicted \ negative}   
\cr 
{{\rm real\ positive }}    &{ \rm true \ positive \ (TP)}         &{\rm false\ negative\ (FN) }       \cr
{\rm real\ negative}    & {\rm false\ positive\ (FP) }          & {\rm true\ negative\ (TN) }       \cr
}.
\end{align*}
\end{small}
\noindent Based on  the confusion matrix, several useful  metrics can be defined, followed by the definitions of G-mean and F-measure: 
\begin{align*}
{\rm recall}&=\frac{\rm TP}{\rm TP+FN}, \quad {\rm specificity}=\frac{\rm TN}{\rm TN+FP}, \quad {\rm precision}=\frac{\rm TP}{\rm TP+FP},\\
{\rm G\!-\!mean}&=\sqrt{{\rm recall}\times {\rm specificity}},\quad 
{\rm F\!-\!measure}=\frac{2\times {\rm precision }\times {\rm recall}}{{\rm precision}+{\rm recall}}.
\end{align*}
Recall and specificity measure the relevance between the predicted and ground-true results for the positive class and negative class, respectively.  On the other hand, precision is the prediction accuracy for the positive class. As seen, G-mean is the geometric mean of recall and specificity, representing the average detection rate of positive and negative classes.  However, one may be only interested in the highly effective detection for the rare case in some applications, e.g.,  cancer detection.  In this case, F-measure is adopted which concerns only the positive class, integrating the detection and prediction rates as a single metric.  

\subsubsection{Experimental Dataset} We consider the learning task of training classifiers using the well-known MNIST dataset of handwritten digits as described in Section~\ref{sec: retransmission needed}.   The training and test sets consist of  $60,000$ and $10,000$ samples, respectively.  Each sample is a grey-valued image of $28 \times 28$ pixels that gives the sample dimensions $p=784$.  For binary classification, we consider both balanced and imbalanced datasets.  For a balanced dataset, we choose the relatively less differentiable class pair of  ``3" and ``5" (according to t-SNE visualization).  For an imbalanced data set, the relatively compact class  ``1" is chosen as the minority class and the majority class is made up of  the remaining classes.     The training set used in experiments is partitioned as follows.  At the edge server, the priorly available collection of clean observations ${\cal L}_0$ are constructed by randomly sampling the global training dataset based on fixed  ratios over classes: a)  $2$ samples for each class for the case of balanced data; b) $1$ sample for minority and $8$ samples for majority for the case of imbalanced data.  The remaining training data are evenly and randomly distributed over edge devices. 
The maximum transmission budget $N$ is set to be $4000$ and $20,000$ (channel uses) for binary and multi-class classification, respectively.  All results are averaged over 200 and 20 experiments for binary and multi-class cases, respectively.

\subsubsection{Learning Model Implementation} The considered classifier models include the following: binary SVM, multi-class SVM, and CNN.  For binary SVM,  the soft-margin SVM is implemented with slack variable set as $1$.   \emph{Iterative Single Data Algorithm} (ISDA) \cite{kecman2005iterative} is used for solving the SVM problem with maximum $10^6$ iterations. The multi-class SVM is built on $45$ binary SVMs as described in Section~\ref{sec: multi-class}.  For the implementation of CNN, we use a  $6$-layer CNN as illustrated in Fig.~\ref{Fig: CNN}, including two $3 \times3$ convolution layers with batch normalization before ReLu activation (the first with $16$ channels, the second with $32$), the first one followed with a $2\times2$ max pooling layer and the second one followed with a fully connected layer, a softmax layer, and a final classification layer.   The model is trained using stochastic gradient descent with momentum~\cite{sutskever2013importance}. The mini-batch size is $2048$, and the number of epochs is $120$.  To accelerate training, the CNN is updated in a batch mode with the incremental sample size as $10$.  
}
\vspace{-4mm}
\subsection{Quality-vs-Quantity Tradeoff}
\begin{figure}[t]
\centering
\subfigure[Channel-aware ARQ]{
\label{Fig: SNR threshold}
\includegraphics[width=6.5cm]{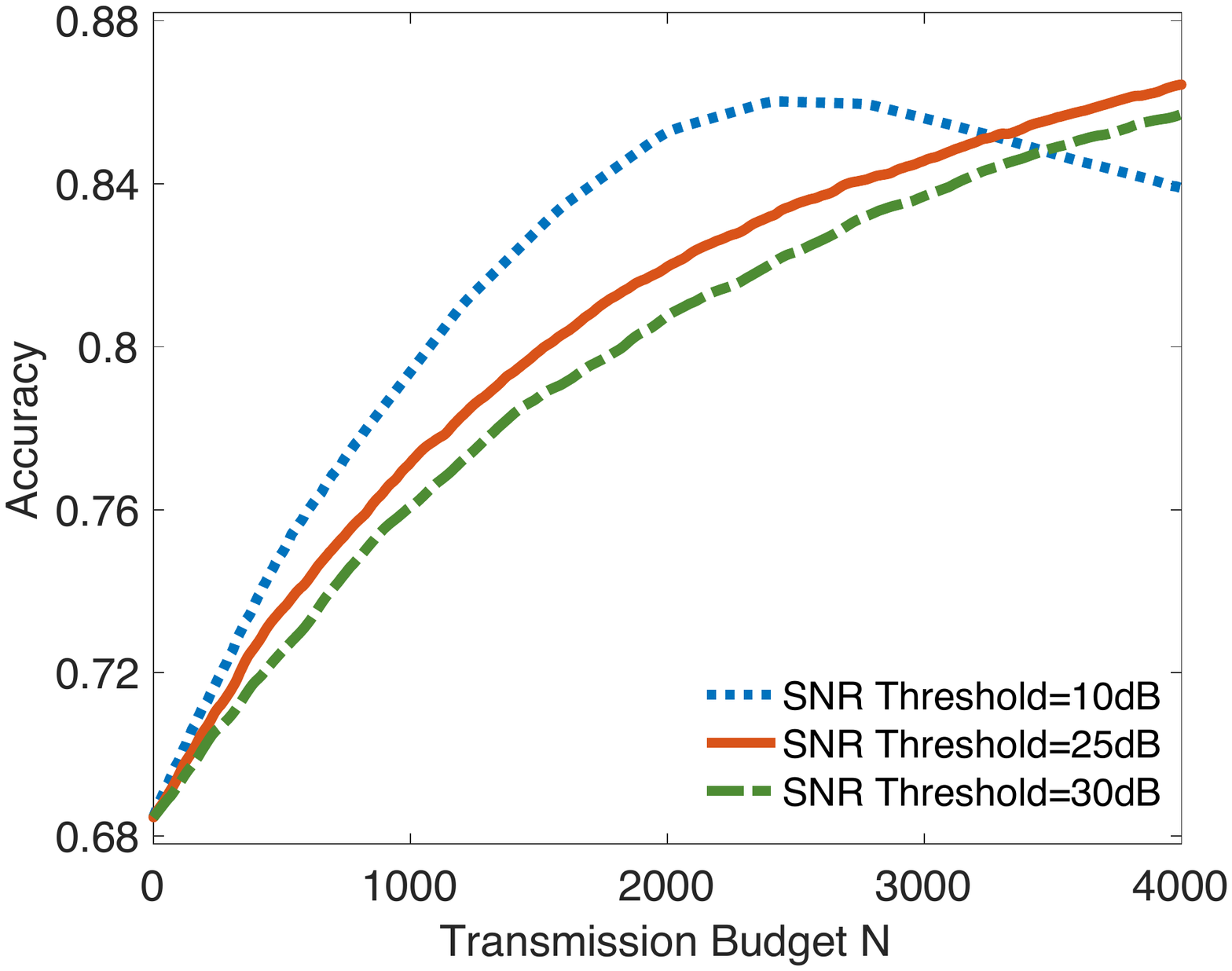}}
\subfigure[Importance ARQ]{
\label{Fig: confidence level}
\includegraphics[width=6.5cm]{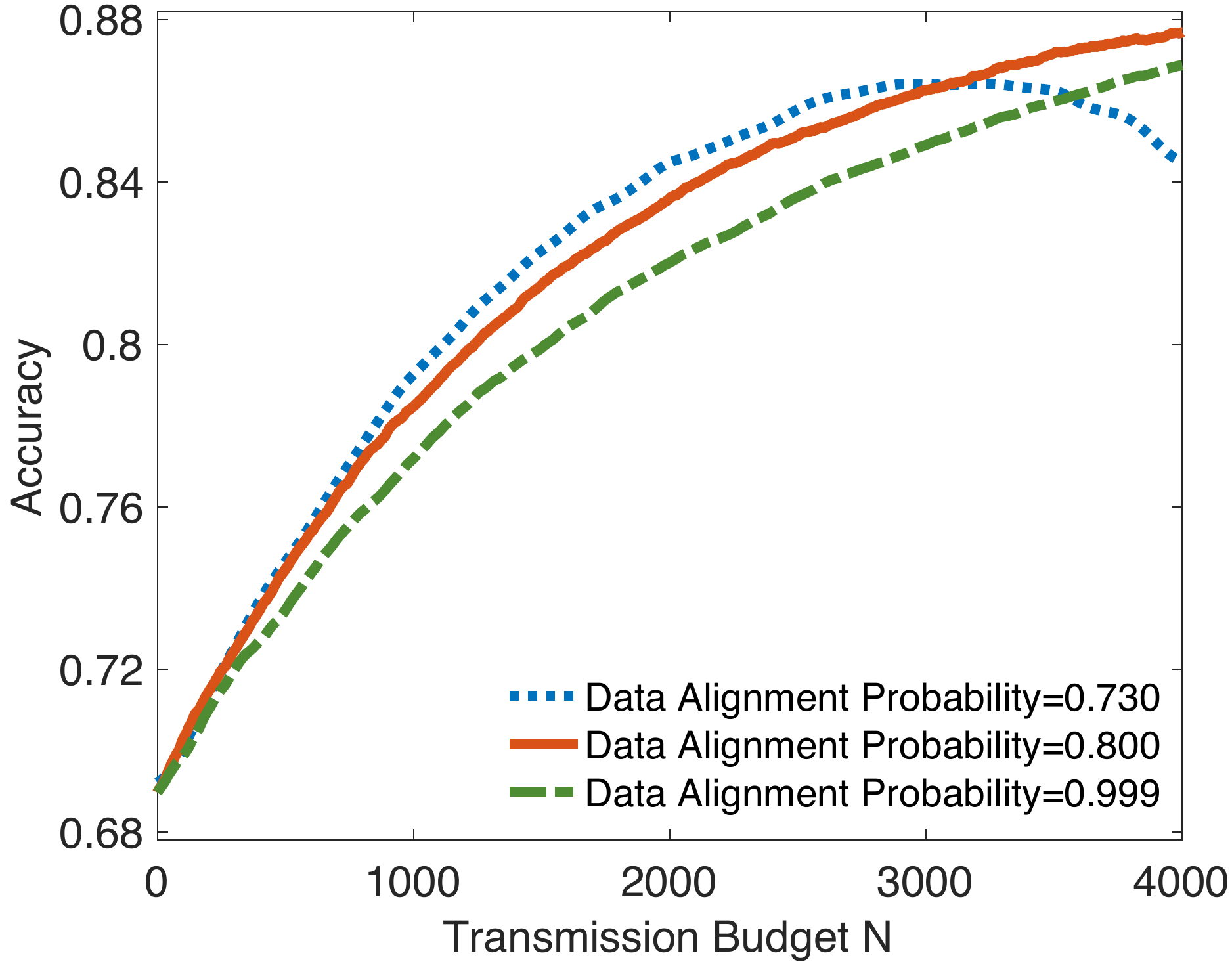}}
\caption{Quality-vs-quantity tradeoff in wireless data acquisition.}
\vspace{-6mm}
\label{Fig: parameters}
\end{figure}

To demonstrate the quality-vs-quantity  tradeoff in wireless data acquisition, 
Fig.~\ref{Fig: parameters} displays the curves of learning accuracy versus transmission budget  for both channel-aware ARQ and importance ARQ.  In Fig.~\ref{Fig: SNR threshold}, we test the performance of channel-aware ARQ with three SNR thresholds, i.e., $\theta_{\sf SNR}=10,\ 25$ and $30$ dB, from low to high data-reliability  requirements.  Similar cases  for importance ARQ  are considered   in Fig.~\ref{Fig: confidence level}  with the reliability  requirements specified by  the data-alignment probability: $p_c=0.730,\ 0.800$ and $0.999$. It is observed from both Fig.~\ref{Fig: SNR threshold} and \ref{Fig: confidence level} that setting the thresholds too low (e.g., $\theta_{\sf SNR}=10$ and $p_c=0.730$) leads to a fast convergence rate but at a cost of performance degradation as the errors accumulate. In contrast, a too high threshold (e.g., $\theta_{\sf SNR}=30$ and $p_c=0.999$) also leads to  poor learning performance due to insufficient acquired  samples.   This suggests that the retransmission threshold should be carefully picked  for optimizing the quality-vs-quantity tradeoff and thereby improving the learning performance. In the following experiments, we select thresholds based on observations in this sub-section to optimize  performance.  

\vspace{-15pt}

\subsection{Learning Performance for Balanced Data}

\begin{figure}[t]
\centering
\subfigure[Learning performance for different values of average transmit SNR~$\bar{\rho}$.]{
\label{Fig: binary performance}
\includegraphics[width=13cm]{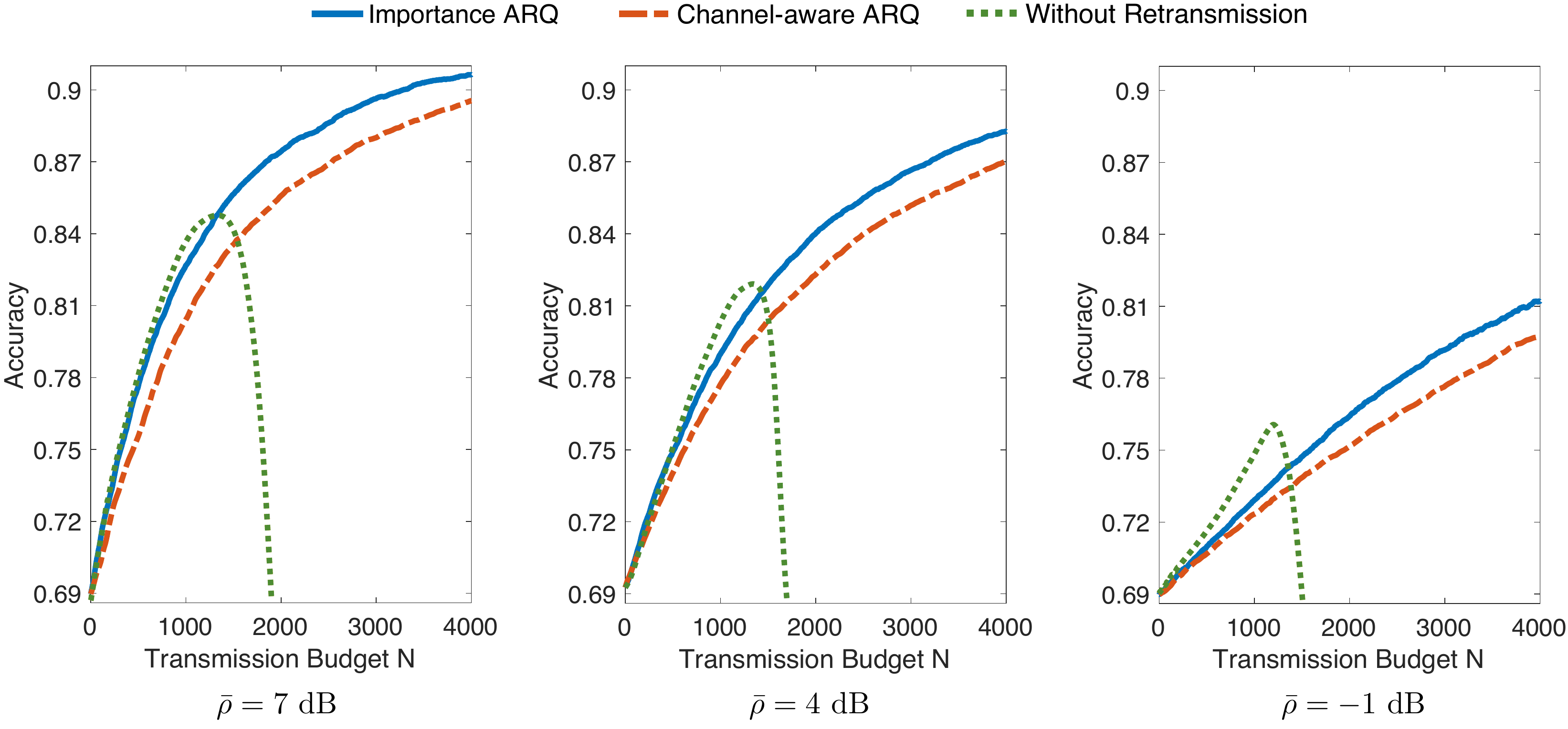}}
\subfigure[Retransmission Distribution]{
\label{Fig: binary retrans}
\includegraphics[width=6cm]{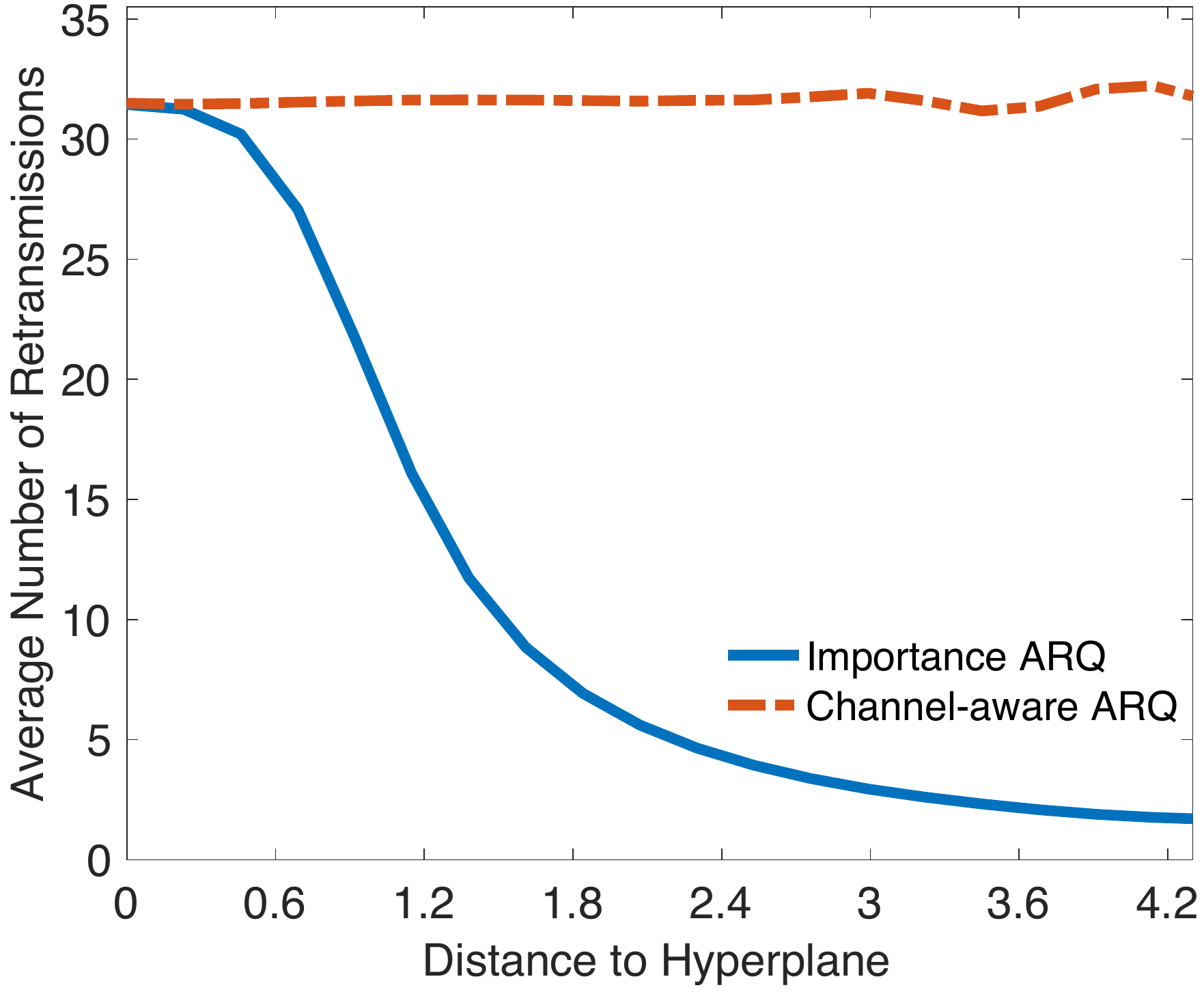}}
\caption{Learning performance for a  binary SVM   classifier trained using with wirelessly acquired  data. }
\vspace{-6mm}
\label{Fig: Binary}
\end{figure}

\subsubsection{Binary SVM Classification} In Fig.~\ref{Fig: Binary}, the learning performance of the proposed importance ARQ is compared with two baseline protocols, namely the channel-aware ARQ and the protocol without retransmission (maximum data quantity). 
It is observed that the performance of edge learning without retransmission dramatically degrades  after acquiring a sufficiently large number of  noisy samples. This is aligned with our previous observations from  Fig. \ref{fig: FixedRetrans} and justifies the need for retransmission. Next, one can observe  that importance ARQ outperforms the conventional channel-aware ARQ throughout the entire training duration. This confirms the performance gain from the intelligent resource utilization in  data acquisition.  Furthermore, the performance gain of importance ARQ  is almost the same in varying SNR scenarios. This demonstrates the robustness of the proposed protocol against the hostile channel condition.

In Fig. \ref{Fig: binary retrans}, we  further investigate the underlying reason for the performance improvement of importance ARQ by plotting the distribution of average numbers of retransmissions over a range of sample uncertainty (inversely proportional to sample distance to the decision hyperplane). One can observe close-to-uniform  distribution for conventional channel-aware ARQ corresponding to uncertainty independence. In contrast, for importance ARQ, retransmission is concentrated in the high uncertainty region. This is aligned with the design principle and shows its effectiveness in adapting retransmission to data importance.

\subsubsection{Multi-class SVM Classification}  In Fig.~\ref{Fig: multicalss perf}, the learning performance of the proposed importance ARQ is compared with two baseline protocols in the scenario of multi-class classification. Similar trends as in the binary-class scenario are observed, and the importance ARQ is found  to consistently outperform the benchmarking protocols in this more challenging scenario. The relation between importance ARQ and the multi-cluster  structure of the training dataset  is illustrated in Fig. \ref{Fig: distr clean importance aware}.  The blue bar indicates that samples in different classes in general have distinct average distances to their corresponding decision hyperplanes and thus with different uncertainty (inverse of distance). From the yellow bar, one can observe that importance ARQ can effectively adapt the average retransmission budget for different classes to their uncertainty levels. For example, class 5 has the shortest average distance to the hyperplane thus is allocated the largest transmission budget to protect its receive quality. In contrast, class 0 has the longest distance, and thus consumes less budget as desired.

\begin{figure}[t]
\centering
\subfigure[Performance evaluation.]
{\label{Fig: multicalss perf}
\includegraphics[width=6.5cm]{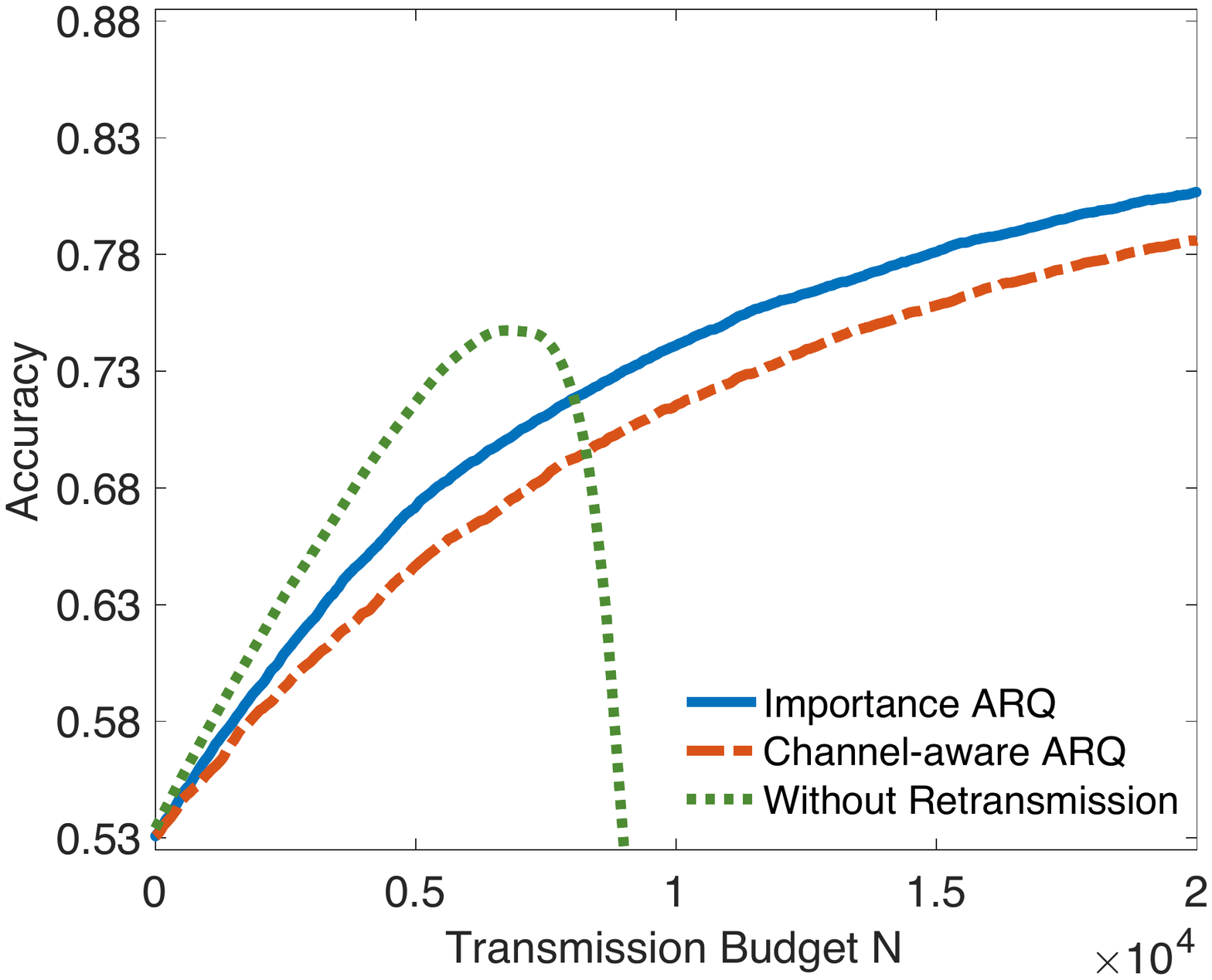}}
\subfigure[Uncertainty v.s. Retransmission.]
{\label{Fig: distr clean importance aware}
\includegraphics[width=7cm]{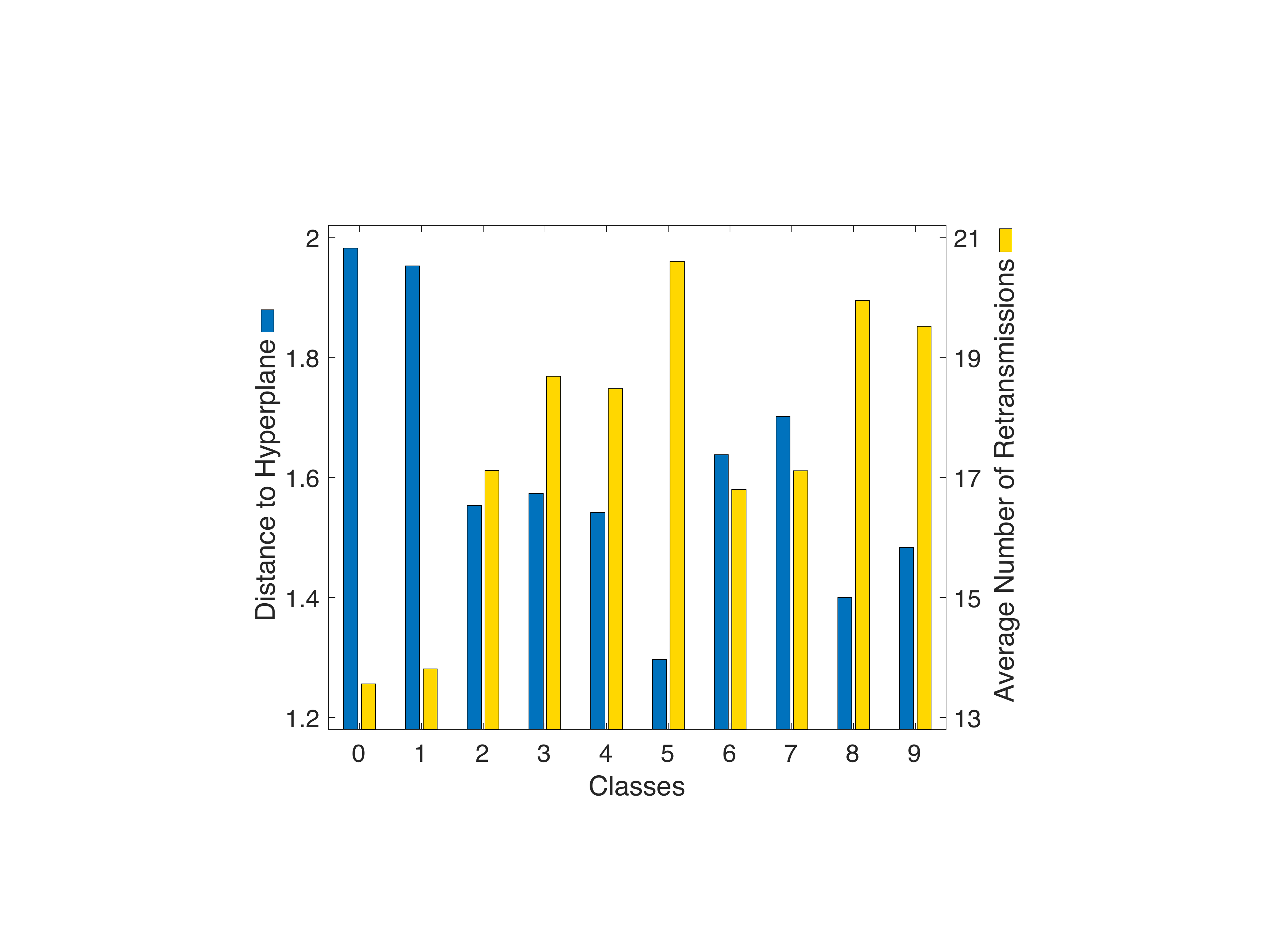}}
\caption{Learning performance evaluation for  multi-class SVM classification.}
\vspace{-6mm}
\label{Fig: Multi-class}
\end{figure}

{\color{black}

\subsubsection{Multi-class Classification of CNN} Our heuristic design for CNN is tested in the scenario of multi-class classification and the related results are provided in Fig.~\ref{Fig: Multi-class CNN}.  Fig.~\ref{Fig: Balanced perf CNN} displays the learning performance adopting entropy based uncertainty, which consistently outperforms two baseline protocols.  One can notice that without retransmission the performance of CNN quickly degrades especially compared with the previous result for SVM, which implicates that CNN is more sensitive to noisy environment.  This is due to the fact that, in CNN classifier, all samples contribute to define the decision hyperplane.  As a result, the noisy effect accumulates faster in CNN than SVM where only a few support vectors determine the decision hyperplane.  Therefore, CNN in general requires more retransmission to attain a higher receive SNR to guarantee the learning performance as shown in Fig.~\ref{Fig: Balanced distr CNN }.  Besides, Fig.~\ref{Fig: Balanced distr CNN } shows the linear relationship between entropy and the number of retransmissions, which is consistent with our design.

\begin{figure}[t]
\centering
\subfigure[Performance evaluation.]
{\label{Fig: Balanced perf CNN}
\includegraphics[width=6.5cm]{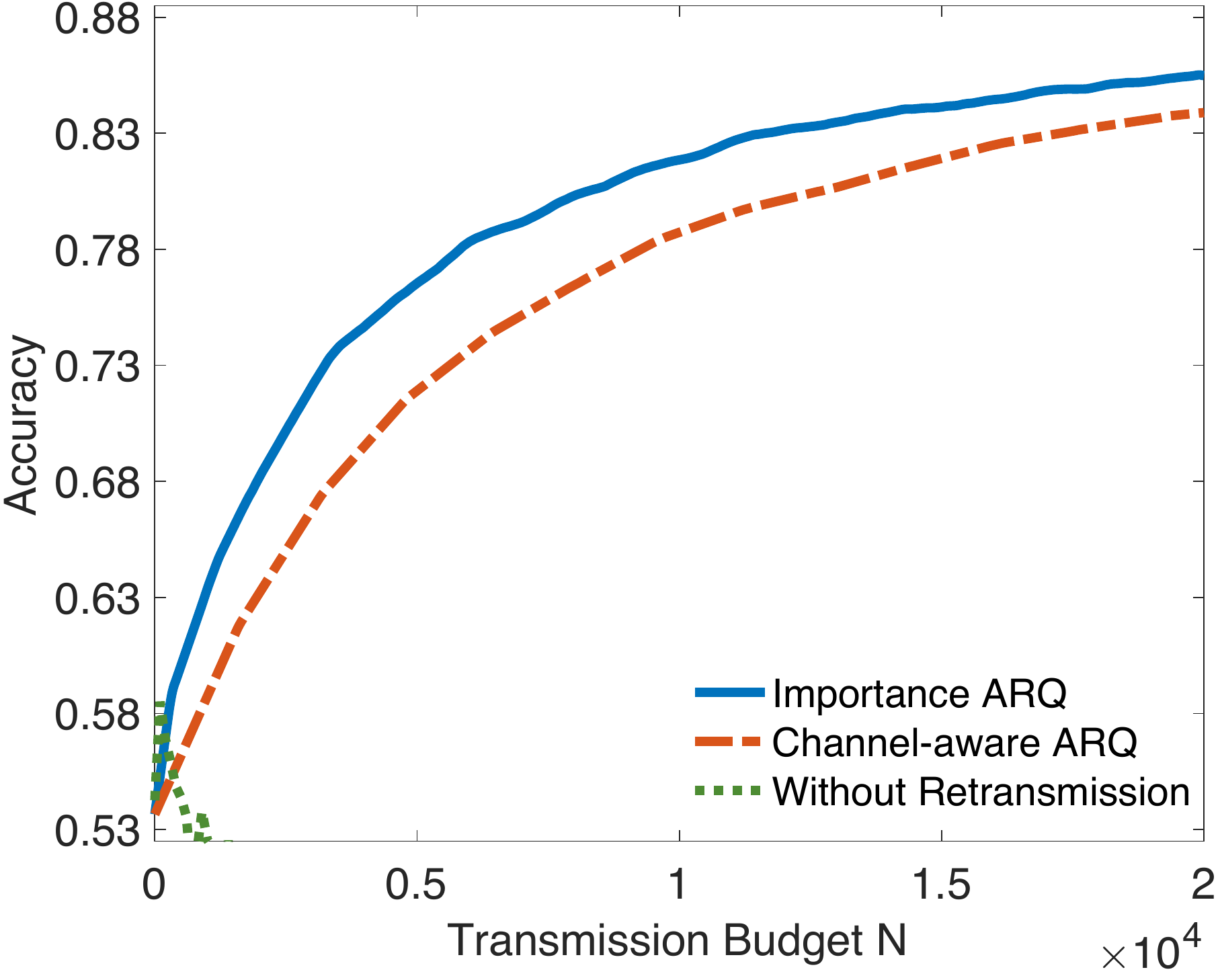}}
\subfigure[Retransmission Distribution.]
{\label{Fig: Balanced distr CNN }
\includegraphics[width=6.3cm]{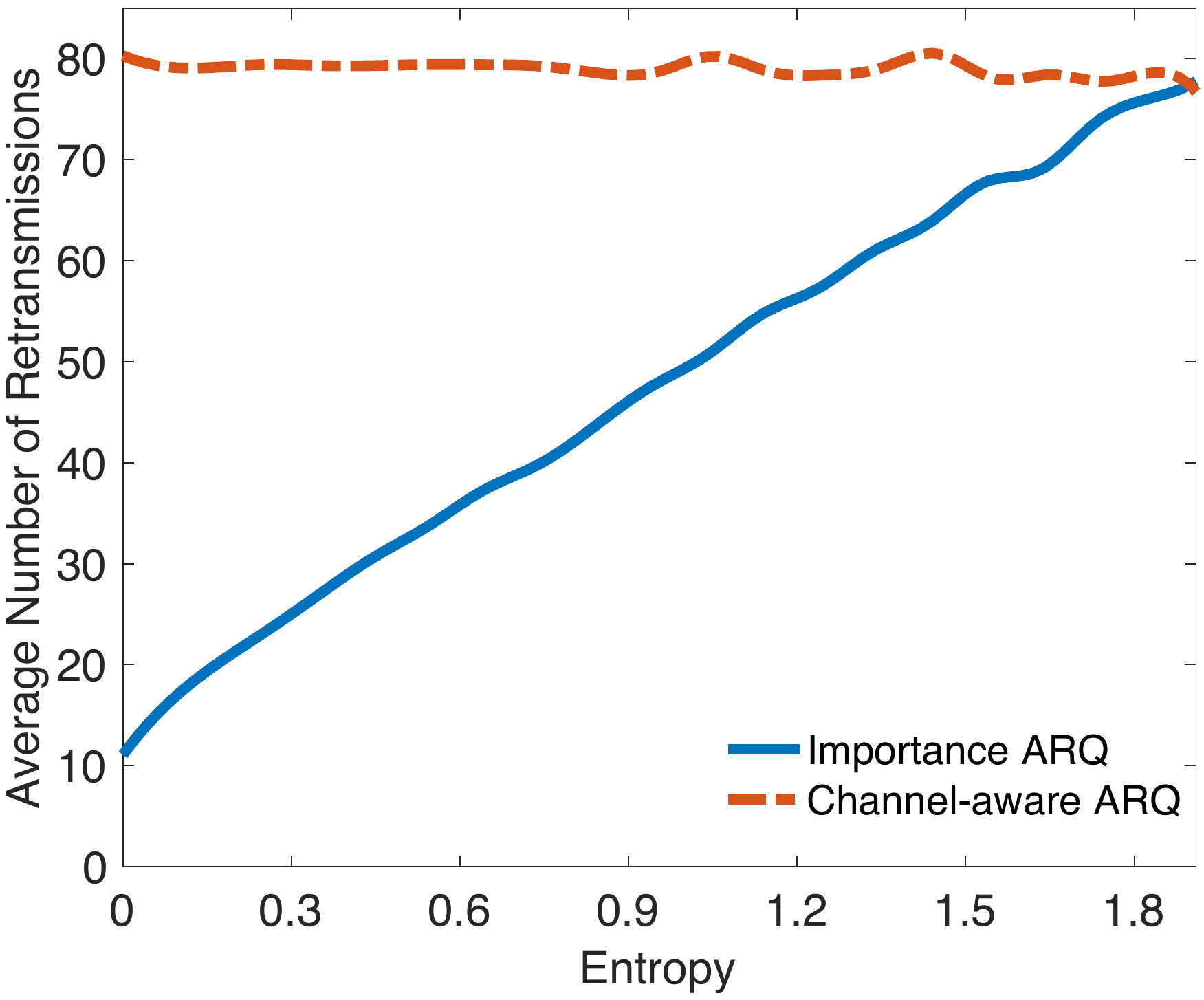}}
\caption{Learning performance evaluation for multi-class CNN classification.}
\vspace{-6mm}
\label{Fig: Multi-class CNN}
\end{figure}

\vspace{-10pt}
\subsection{Learning Performance for Imbalanced Data}

\subsubsection{Imbalanced Classification of SVM} In Fig.~\ref{Fig: imbalanced performance SVM}, both F-measure and G-mean of the proposed importance ARQ are compared with two baseline protocols in the scenario of imbalanced classification by using SVM. Compared with balanced classification (Fig.~\ref{Fig: binary performance}), the performance curves in the imbalanced setting degrade faster if no retransmission is made, which implicates that imbalanced classification is more vulnerable to the hostile channel environment.  This fact calls for an intelligent retransmission protocol to regulate the quality of each training sample.  One can notice that  importance ARQ could achieve a larger gain in imbalanced classification (nearly 10$\%$ performance improvement is observed compared with the conventional channel-aware ARQ).  To further investigate the underpinning reason, we visualize the imbalanced dataset by using t-SNE and plot the relationship between retransmission and uncertainty.  The left subfigure in Fig.~\ref{Fig: imbalanced retrans} shows that, the minority class has a higher uncertainty value since the average distance to the hyperplane is shorter than the majority one.  This is aligned with the blue bar in the right subfigure.  It is also observed that a highly uncertain minority class consumes more retransmission budget in importance ARQ, as shown by the red bar.  However, the green bar  shows that channel-aware ARQ allocates equal transmission budgets to both majority and minority classes.  The superiority of importance ARQ in the balanced setting further substantiates the theoretical gain brought by the intelligent adaptation of the radio resource allocation according to the data importance.  

\begin{figure}[t]
\centering
\subfigure[Learning performance measured by F-measure and G-mean.]{
\label{Fig: imbalanced performance SVM}
\includegraphics[width=13.4cm]{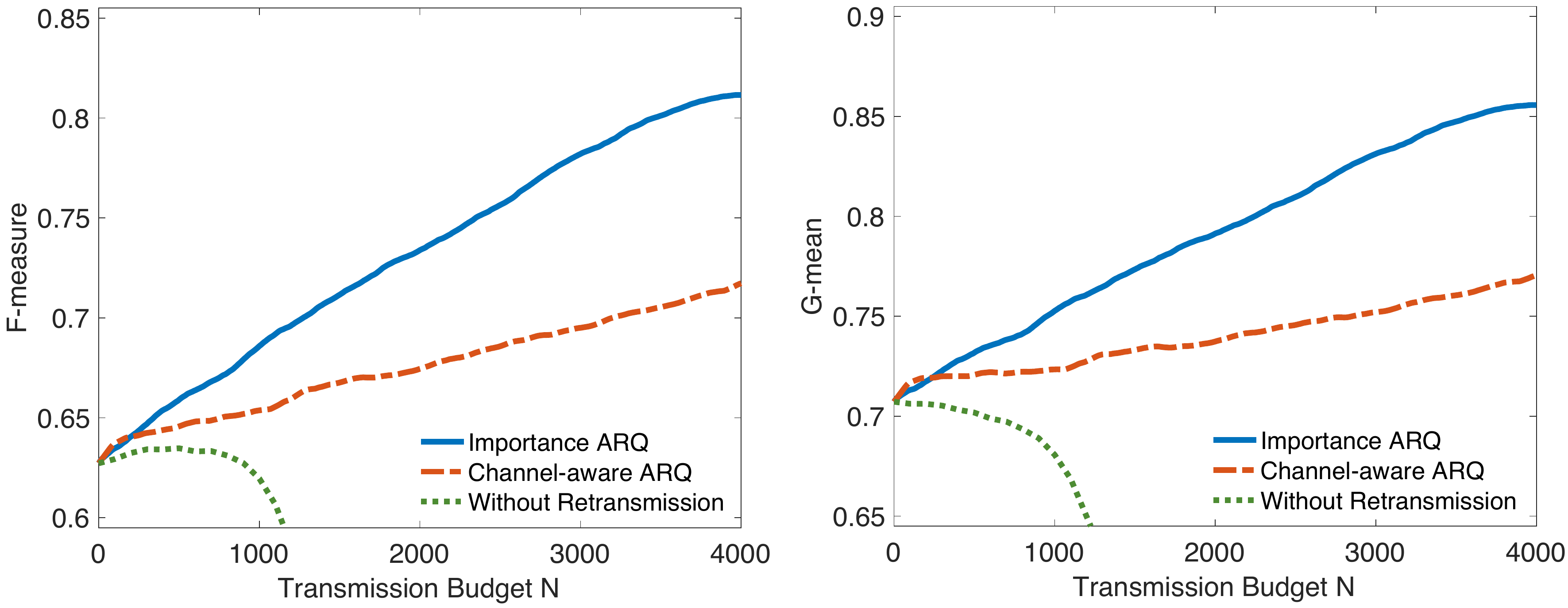}}
\subfigure[Uncertainty v.s. Retransmission.]{
\label{Fig: imbalanced retrans}
\includegraphics[width=16cm]{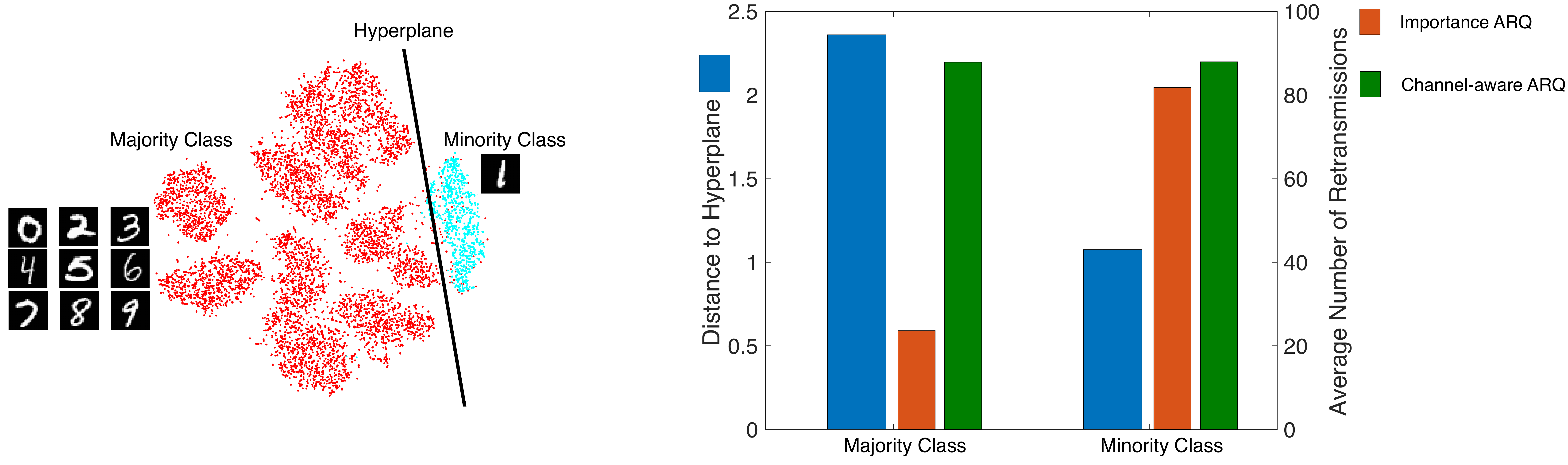}
}
\caption{Learning performance for SVM classifier training using  imbalanced data. }
\vspace{-6mm}
\label{Fig: Binary Imbalanced}
\end{figure}

\subsubsection{Imbalanced Classification of CNN} In Fig.~\ref{Fig: Imbalanced CNN}, F-measure and G-mean are examined in the imbalanced classification by deploying CNN.  Similar trends as the SVM classifier are observed, and the importance ARQ is found to consistently outperform the benchmarking protocols, which confirm  the effectiveness of our extension in Section~\ref{sec: extension}.

\begin{figure}[t]
\begin{center}
{\includegraphics[width=13.4cm]{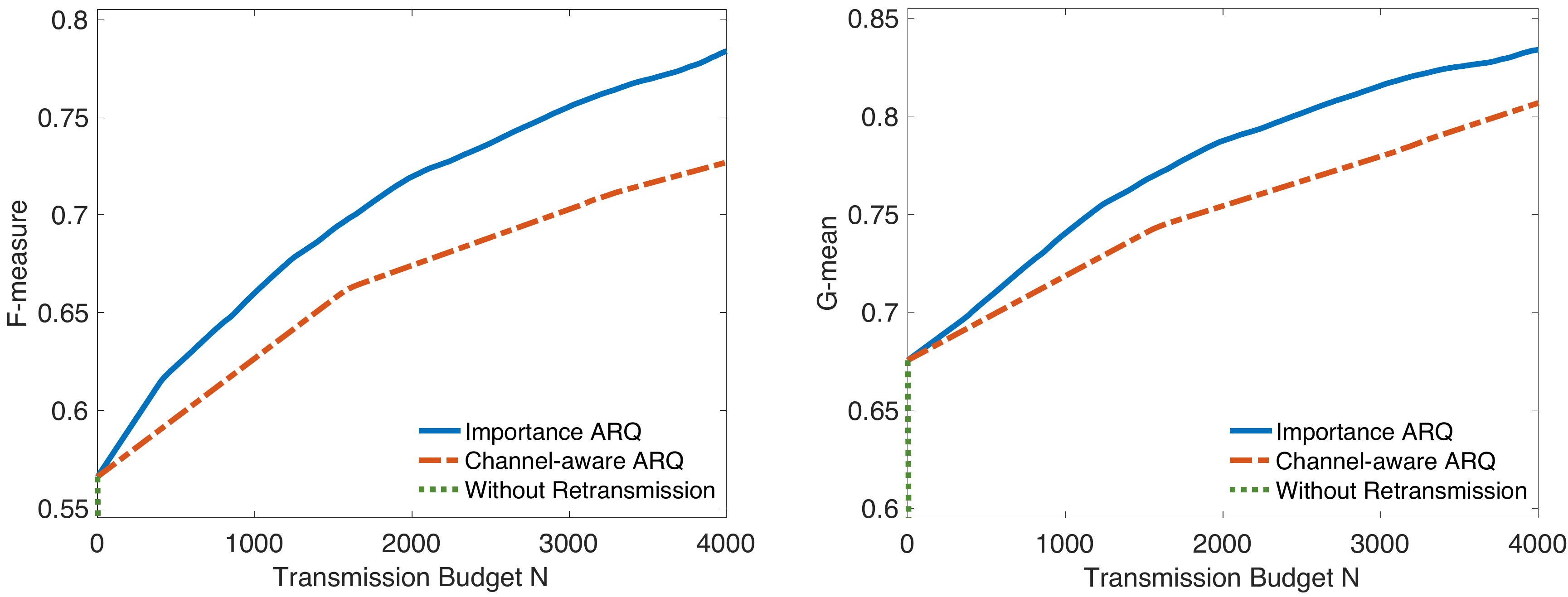}}\vspace{-13pt}
\caption{Learning performance for CNN classifier training using  imbalanced data.}\vspace{-10pt}
\label{Fig: Imbalanced CNN}
\vspace{-6mm}
\end{center}
\end{figure}

}
\vspace{-10pt}

\section{Concluding Remarks}\label{sec: concluding remarks}
In this paper, we have proposed a novel retransmission protocol, namely importance ARQ, for wireless data acquisition in edge learning systems. It intelligently  adapts retransmission to data-sample  importance so as to enhance the learning performance given a transmission latency constraint. Comprehensive experiments using  real datasets  substantiate the performance gain of the proposed design. 

At a higher level, the work contributes  the new principle of exploiting the non-uniform distribution of data importance to improve the efficiency of wireless data acquisition for edge learning. Importance aware retransmission is just one way  for materializing this principle. It can be applied to many other aspects of wireless data acquisition such as scheduling, power control, spectrum allocation, and energy efficient transmission. Thereby many promising research opportunities are presented. 

\vspace{-10pt}

\appendices 
\section{Proof of Lemma~\ref{lemma: noiseless data distribution}} \label{app: noiseless data distribution}
\vspace{-8pt}

The  received sample $\hat{{ \bf x}}{(T)}$  in \eqref{eq: est x} can be  rewritten as 
\vspace{-8pt}
\begin{align}
\hat{{ \bf x}}{(T)}&={ \bf x}+\Re \l(\widetilde{\bf z}(T)\r), \nn
\end{align}
\vspace{-8pt}
where $\widetilde{\bf z}(T) = \frac{1}{ \sqrt{P}}\l(\frac{\sum_{i = n+1 }^{n+T}\l(h^{(i)}\r)^*{{\bf z}^{(i)}}}{\sum_{m = n+1 }^{n+T}|h^{(m)}|^2}\r)$.
Consequently, the transmitted sample is 
\begin{equation}
{{ \bf x}}=\hat{{ \bf x}}{(T)}-\Re \l(\widetilde{\bf z}{(T)}\r), \vspace{-8pt}
\end{equation}
where $\widetilde{{\bf z}}{(T)}=\l[\widetilde{z}_1{(T)}, \cdots, \widetilde{z}_p{(T)}\r]^{\sf T}$ is the equivalent noise after combining with the entries being  
\begin{equation}\label{eq: expression of equivalent noise}
\widetilde{z}_j{(T)}=\frac{1}{\sqrt{P}}\times\frac{\sum_{i=n+1}^{n+T}\l(h^{(i)}\r)^*z_j^{(i)}}{\sum_{m = n+1 }^{n+T}|h^{(m)}|^2},\ j=1,2,\cdots, p. \nn
\end{equation}
Since $z_j^{(i)}$ follows i.i.d ${\cal CN}\l(0,\sigma^2\r)$, each entries in $\widetilde{{\bf z}}{(T)}$ are i.i.d and the distributions are:
\begin{equation*}
\widetilde{z}_j{(T)}\sim{\cal CN}\l(0,\frac{\sigma^2}{\sum_{i = n+1 }^{n+T}|h^{(i)}|^2P}\r), \ j=1,2,\cdots, p.
\end{equation*}
With effective SNR defined in \eqref{eq: def SNR}, taking the real part of $\widetilde{\bf z}$ yields to the following distribution:
\begin{equation*}
\Re \l(\widetilde{\bf z}{(T)}\r)\sim{\cal N}\l(0,\frac{\bf I}{\SNR(T)} \r),
\end{equation*}
which leads to the desired result in \eqref{eq:data distribution}.

%

\bibliographystyle{ieeetr}

\end{document}